\documentclass[11pt,letterpaper]{article}
\pdfoutput=1

\usepackage[dvipsnames,usenames]{xcolor}
\PassOptionsToPackage{hyphens}{url}
\usepackage[colorlinks=true,urlcolor=Blue,citecolor=Green,linkcolor=BrickRed]{hyperref}
\usepackage{amsmath,amsthm,amssymb,amsfonts}
\usepackage[utf8]{inputenc}
\usepackage{todonotes}
\usepackage{xspace}
\usepackage{multicol}
\usepackage{thmtools,thm-restate}
\usepackage{authblk}
\usepackage{multirow}
\usepackage{algorithm}
\usepackage{algorithmicx}
\usepackage[noend]{algpseudocode}
\usepackage{tikz}
\usepackage[noadjust]{cite}
\usepackage{cellspace}

\usepackage[left=1in,right=1in,top=1in,bottom=1in]{geometry}

\newtheorem{theorem}{Theorem}[section]
\newtheorem{lemma}[theorem]{Lemma}
\newtheorem{corollary}[theorem]{Corollary}
\newtheorem{definition}[lemma]{Definition}
\newtheorem{observation}[lemma]{Observation}
\newtheorem{proposition}[lemma]{Proposition}
\newtheorem{fact}[lemma]{Fact}

\newtheorem{question}[lemma]{Question}

\newtheorem{def-restatable}[theorem]{Definition}
\newtheorem{obs-restatable}{Lemma}
\newtheorem{lemma-restatable}[theorem]{Lemma}
\newtheorem{corollary-restatable}[theorem]{Corollary}

\newcommand{\eps}{\varepsilon}
\newcommand{\Oh}{\mathcal{O}}
\newcommand{\Ohtilda}{\tilde{\Oh}}
\newcommand{\red}{\leq_2}
\newcommand{\reddesc}[1]{\stackrel{\text{#1} }\red}
\newcommand{\eqq}{\equiv_2}
\newcommand{\eqqdesc}[1]{\stackrel{\text{#1} }\eqq}
\newcommand{\ldt}[3]{\text{3\LDT}(#1,\overline{#2},#3)}
\newcommand{\ldtc}[3]{\text{3\LDT}_c(#1,\overline{#2},#3)}
\newcommand{\norm}[1]{\|#1\|}
\newcommand{\Patrascu}{P\v{a}tra\c{s}cu}
\newcommand{\ind}{[-n', n']}
\DeclareMathOperator{\sgn}{sgn}

\def\polylog{\operatorname{polylog}}

\DeclareMathOperator{\lcm}{lcm}

\newcommand{\SUM}{\textsf{SUM}}
\newcommand{\LDT}{\textsf{LDT}}
\newcommand{\AVG}{\textsf{Average}}
\newcommand{\Convolution}{\textsf{Conv}}
\newcommand{\Conv}{\Convolution3\LDT}

\newcommand{\II}{\mathcal{I}}
\newcommand{\BB}{\mathcal{B}}

   \newcommand{\defproblemwithparameters}[4]{
  \vspace{2mm}
\noindent\fbox{
  \begin{minipage}{0.96\textwidth}
  #1\\
  {\bf{Parameters:}} #4.\\
  {\bf{Input:}} #2  \\
  {\bf{Output:}} Are there #3?
  \end{minipage}
  }
  \vspace{2mm}
}

\title{Equivalences between Non-trivial Variants of 3LDT and Conv3LDT\footnote{A preliminary version appeared in Proceedings of the 52nd Annual ACM SIGACT Symposium on Theory of Computing (STOC 2020), pp. 974--981 \cite{DudekGS20}. Partially supported by the Polish National Science Centre grant number 2023/51/B/ST6/01505.}}
\date{}

\author[1]{Bartłomiej Dudek}

\author[1]{Paweł Gawrychowski}

\author[2]{Tatiana Starikovskaya}

\affil[1]{Institute of Computer Science, University of Wrocław, Poland}
  \affil[ ]{\texttt{\{bartlomiej.dudek,gawry\}@cs.uni.wroc.pl}}
\affil[2]{DIENS, \'{E}cole normale sup\'{e}rieure, PSL Research University, France}
\affil[ ]{\texttt{tat.starikovskaya@gmail.com}}

\begin{document}
\maketitle

\begin{abstract}
The popular 3\SUM{} conjecture states that there is no strongly subquadratic time algorithm for checking if a given set of integers contains three distinct elements $x_1, x_2, x_3$ such that $x_1+x_2=x_3$. A closely related problem is to check if a given set of integers contains distinct elements satisfying $x_1+x_2=2x_3$. This can be reduced to 3\SUM{} in almost-linear time, but surprisingly a reverse reduction establishing 3\SUM{} hardness was not known.

We provide such a reduction, thus resolving an open question of Erickson. In fact, we consider a more general problem called 3\LDT{} parameterized by integer parameters $\alpha_1, \alpha_2, \alpha_3$ and $t$. In this problem, we need to check if a given set of  integers contains distinct elements $x_1, x_2, x_3$ such that $\alpha_1 x_1+\alpha_2 x_2 +\alpha_3 x_3 = t$. For some combinations of the parameters, every instance of this problem is a NO-instance or there exists a simple almost-linear time algorithm. We call such variants trivial.
We prove that all non-trivial variants of 3\LDT{} over the same universe $[-n^c,n^c]$ for some $c\geq2$ are equivalent under subquadratic reductions.
The main technical tool used in our proof is an application of the famous Behrend's construction that partitions a given set of integers into few subsets that avoid a chosen linear equation.

\Patrascu{} observed that while proving 3\SUM{}-hardness it is convenient to work with a more structured problem related to 3\SUM{}, \Convolution3\SUM{}.
Hence it is natural to consider also Convolution variants of 3\LDT{} in which we are given an array of numbers and need to find three elements such that both their values and indices satisfy the linear equation, that is: $\sum_{i=1}^3 \alpha_i j_i=t$ and  $\sum_{i=1}^3 \alpha_i A[j_i]=t$.
We extend our results to \Convolution3\LDT{} and show that for all $c\geq2$, all non-trivial variants of 3\LDT{} over the universe $[-n^c,n^c]$ and of \Convolution3\LDT{} over the universe $[-n^{c-1},n^{c-1}]$ are subquadratic-equivalent, so in particular they are all equivalent to 3\SUM{} under subquadratic reductions.

As the last step, we show how to apply the methods of Fischer, Kaliciak, and Polak [ITCS 2024] to show that we can reduce non-trivial variant of 3\LDT{} over an arbitrary universe to the same variant over cubic universe, and similarly reduce \Conv{} to instances over quadratic universe. This gives us a complete classification of variants of 3\LDT{} and \Conv: every variant is either subquadratically equivalent to 3\SUM{} over $[-n^c,n^c]$ for some $c\in[2,3]$ or can be solved in subquadratic time.
\end{abstract}

\newpage

\section{Introduction}\label{sec:introduction}

The well-known 3\SUM{} problem is to decide, given a set $X$ of $n$ integers, whether any three distinct elements $x_1,x_2,x_3$ of $X$ satisfy $x_1+x_2=x_3$.
This can be easily solved in quadratic time by first sorting $X$, checking all candidates for $x_3$ and for each of them scanning the sorted sequence with two pointers.
For many years no faster algorithm was known, and it was conjectured that no significantly faster algorithm exists.
This assumption led to strong lower bounds for multiple problems in computational geometry~\cite{GajentaanO95}
and, more recently, became a central problem in the field of fine-grained complexity \cite{Williams15}.
Furthermore, it has been proven that in some restricted models of computation 3\SUM{} requires
$\Omega(n^{2})$ time~\cite{EricksonBounds99,AilonC05}.

However, in 2014 Gr{\o}nlund and Pettie~\cite{GronlundP18} proved that the decision tree complexity of 3\SUM{} is only $\Oh(n^{1.5}\sqrt{\log n})$, which ruled out any almost quadratic-time lower bounds in the decision tree model. This was later improved by Gold and Sharir to $\Oh(n^{1.5})$~\cite{GoldS17} and finally to $\Oh(n\log^2 n)$ by Kane et al. \cite{KaneLM19}.
The upper bounds for the decision tree model were later used to design a series of algorithms for a version of 3\SUM{} in the real RAM model where the set $X$ can also contain real numbers.
In this model, Gr{\o}nlund and Pettie~\cite{GronlundP18} derived an $\Oh(n^{2}(\log\log n)^{2}/\log n)$ time randomized algorithm and an $\Oh(n^{2}(\log\log n/\log n)^{2/3})$ time deterministic algorithm. The best deterministic bound was soon improved to $\Oh(n^{2}\log \log n/\log n)$ by Gold and Sharir~\cite{GoldS17} and (independently) by Freund~\cite{Freund17} and then to $\Oh(n^2(\log\log n)^{\Oh(1)}/\log^2n)$ by Chan \cite{Chan18}. 
These results immediately imply similar bounds for the integer version of 3\SUM{}.
In the word RAM model with machine words of size $w$, Baran, Demaine and \Patrascu\ \cite{BaranDP08} provided an algorithm with $\Oh(n^2/\max\{\frac{w}{\log^2w},\frac{\log^2n}{(\log\log n)^2}\})$ expected time.

Even though asymptotically faster than $\Oh(n^{2})$, these algorithms are not strongly subquadratic, meaning working in $\Oh(n^{2-\eps})$ time, for some $\eps>0$. This motivates the popular modern version of the conjecture, which is that the 3\SUM{} problem cannot be solved in $\Oh(n^{2-\eps})$ time (even in expectation), for any $\eps>0$, on the word RAM model with words of size $\Oh(\log n)$~\cite{Patrascu10}. By now we have multiple examples of other problems that can be shown to be hard assuming this conjecture, especially in geometry~\cite{GajentaanO95,BergGO97a,Erickson99c,ArchambaultEK05,SossEO03,CheongEH04,AbellanasHIKLMPS01,AronovH08,ArkinCHMSSY98,BarequetH01,EricksonHM06,BoseKT98}, but also in also in dynamic algorithms and data structures \cite{AbboudW14,KopelowitzPP16,Patrascu10,Dahlgaard16,AbboudWY18}, string algorithms \cite{AbboudWW14,AmirCLL14,KopelowitzPP16}, finding exact weight subgraphs \cite{WilliamsW13,AbboudL13} and finally in partial matrix multiplication and reporting variants of convolution \cite{GoldsteinKLP16}.

In particular, it is well-known that the 3\SUM{} problem defined above is subquadratic-equivalent to its 3-partite variant, where we are given three sets $S_1, S_2, S_3$ containing at most $n$ integers each, and must decide whether there is $x_1 \in S_1$, $x_2 \in S_2$, and $x_3 \in S_3$ such that $x_1 + x_2 = x_3$. To reduce 3-partite 3\SUM{} to 1-partite, we can add a multiple of some sufficiently big number $M$ to all elements in every set and take the union, for example: 
$$X= \{3M+x: x\in S_1\} \cup \{M+x: x\in S_2\} \cup \{4M+x: x\in S_3\}.$$
$M$ is chosen so that the only possibility for the three elements of $X$ to satisfy $x_1+x_2=x_3$ is that they correspond to three elements belonging to distinct sets $S_1$, $S_2$, and $S_3$. To show the reduction from $1$-partite 3\SUM{} to 3-partite, a natural approach would be to take $S_1 = S_2 = S_3 = X$. However, this does not quite work as in the $1$-partite variant we desire $x_1, x_2, x_3$ to be \emph{distinct}. In the folklore reduction, this technicality is overcome using the so-called color-coding technique by Alon et al.~\cite{AlonYZ95}.

A natural generalization of 3\SUM{} is $3$-variate linear degeneracy testing, or $3$\LDT{} \cite{AilonC05}. In this problem, we are given a set $X$ of $n$ integers, integer parameters $\alpha_1, \alpha_2, \alpha_3$ and~$t$, and must decide whether there are $3$ distinct numbers $x_{1}, x_{2}, x_{3} \in X$ such that $\sum_{i=1}^{3}\alpha_i x_i = t$. Similar to 3\SUM{}, the $3$\LDT{} problem can be considered in the $3$-partite variant as well. 

A particularly natural variant of the $1$-partite $3$\LDT{} problem is as follows: given a set $X$ of $n$ numbers, are there three distinct $x_1, x_2, x_3\in X$
such that $x_1+x_2-2x_3=0$?
In other words, we want to check if a set contains three distinct elements that form an arithmetic progression.
Following Erickson~\cite{Erickson99}, we call this problem \AVG.
It is easy to see that \AVG{} reduces to $\Oh(\log n)$ instances of 3-partite 3\SUM{} where the $j$-th instance consists of the sets $S_j, X\setminus S_j, Y$ such that $S_j=\{x_{i}\in X: \text{ the $j$-th bit of $i$ is 1}\}$ (when $X=\{x_{1},\ldots,x_{n}\}$) and $Y=\{2x: x\in X\}$.
However, a reverse reduction seems more elusive and in fact according to Erickson~\cite{Erickson99} it is not known whether \AVG{} is 3\SUM-hard\footnote{See \url{https://cs.stackexchange.com/questions/10681/is-detecting-doubly-arithmetic-progressions-3sum-hard/10725\#10725}.}. This suggests the following question.

\begin{question}
\label{q:avg}
Can we design a reduction from 3\SUM{} to \AVG{}? Or is \AVG{} easier than 3\SUM{}?
\end{question}

A more ambitious question would be to provide a complete characterisation of all variants of 3\LDT{} parameterized by $\alpha_{1},\alpha_{2},\alpha_{3}, t$. We know that in the restricted 3-linear decision tree model solving every variant where all coefficients $\alpha_i$ are nonzero requires quadratic time \cite{EricksonBounds99,AilonC05}, but by now we know that this model is not necessarily the most appropriate for such problems. Because of that, we focus on the classical word RAM model, which is also the model under which most problems that are hard conditional on the 3\SUM{} conjecture have been studied.

\begin{question}
\label{q:ldt}	
Which variants of 3\LDT{} are easier than others in the word RAM model? Or are they all equivalent?
\end{question}

In this work, we show that all non-trivial variants of 3\LDT{} are subquadratic-equivalent in the word RAM model, which, in particular, implies that \AVG{} is subquadratic-equivalent to 3\SUM{}. Thus, we completely resolve both Question~\ref{q:avg} and Question~\ref{q:ldt}.

In his seminal work~\cite{Patrascu10}, \Patrascu\
introduced a more structured variant of the 3\SUM{} problem called  \Convolution3\SUM.\footnote{According to \cite{ChanH20}, the earliest reference to a problem similar to \Convolution 3\SUM{} is from 2005: \url{https://3dpancakes.typepad.com/ernie/2005/08/easy_but_not_th.html}}
In this problem, we are given an integer array $A$ and must decide whether there exist two distinct indices $i,j$ such that $A[i]+A[j]=A[i+j]$.
This variant is more useful for establishing reductions from 3\SUM{} and led to a number of 3\SUM-hardness results, e.g. for dynamic problems \cite{Patrascu10,AbboudW14,KopelowitzPP16} and string algorithms \cite{AmirCLL14,KopelowitzPP16}.
\Patrascu{}~\cite{Patrascu10} showed that \Convolution3\SUM{} is subquadratic-equivalent to 3\SUM.
Recently, Kopelowitz et al.~\cite{KopelowitzPP16} refined the reduction of \Patrascu{} to achieve a smaller slowdown factor, and Chan and He~\cite{ChanH20} presented a simple deterministic reduction.
Hence it is natural to consider also \Conv{} variants of 3\LDT, in which both indices of the elements and their values need to satisfy the given linear equation.
We further extend our techniques to show that all non-trivial variants of \Convolution3\LDT{} are subquadratic-equivalent to each other and to non-trivial variants of 3\LDT{}, with the size of the universe increased by the factor of~$n$ while switching from convolution to non-convolution variants of 3\LDT.

We provide formal definitions and statements of our results in Section~\ref{sec:overview}.

\paragraph{Follow-up work.}
Using divide and conquer paradigm and techniques from the conference version of our paper, Radoszewski et al. \cite{RadoszewskiRSWZ21} showed that \Convolution\AVG{} is 3\SUM{}-hard (though increasing the size of the universe by the factor of $n$ in the reduction).
This allowed them to prove 3\SUM{}-hardness of various problems related to Abelian squares in a text.

A direct attempt at generalizing our results to show the equivalence between $k$\SUM{} and $k$\LDT{} for $k>3$ seem problematic.
For concreteness, consider the variant of 4\LDT{} in which we ask about existence of 4 distinct numbers $x_1,x_2,x_3,x_4$ such that $x_1+x_2=x_3+x_4$.
The sets avoiding existence of such four numbers are called Sidon sets and are well-studied in additive combinatorics \cite{Obryant04}.
It is easy to see that any Sidon subset of $[N]$ contains at most $\Oh(\sqrt N)$ elements.
However, to apply our approach we need to partition an arbitrary set into few such subsets, which is impossible.
Therefore another idea is required to prove 4\SUM-hardness of 4\LDT.
In their very recent result, Jin and Xu \cite{JinX23} found a different approach and showed that all non-trivial variants of 4\LDT{} are 3\SUM-hard.
As an intermediate step, they showed 3\SUM-hardness of 3\SUM{} on  Sidon sets.
This required an analysis of the additive energy of a given 3\SUM{} instance (defined as the number of quadruples $(a,b,c,d)\in X^4$ such that $a+b=c+d$): for high additive energy they applied Balog-Szemer{\'e}di-Gowers Theorem~\cite {BalogS94} and for smaller additive energy they applied a self reduction based on a very non-trivial hashing.
3\SUM-hardness of 3\SUM{} on Sidon sets also allowed them to establish fascinating new 3\SUM-hardness results for various graph problems, like 4-Cycle Enumeration.
Similar conclusions also follow from the work by Abboud et al. \cite{AbboudBF23} that appeared at the same time as \cite{JinX23}.

Fischer et al. \cite{FischerKP24} showed a deterministic reduction of the universe size for 3\SUM{}, from arbitrarily big to cubic.
We show how to extend their result to all 1- and 3-partite variants of 3\LDT{} and \Conv.

\paragraph{Related work.}
Multiple variants of 3\SUM{} have been also considered, e.g. clustered 3\SUM{} and 3\SUM{} for monotone sets in 2D that are surprisingly solvable in truly subquadratic time~\cite{ChanL15}; algebraic 3\SUM{}, a generalization which replaces the sum function by a constant-degree polynomial~\cite{BarbaCILOS19}; and 3\SUM{}$^+$ in which, given three sets $A,B,S$ one needs to return $(A+B)\cap S$~\cite{GronlundP18,BaranDP08}.
An interesting generalisation of the 3\SUM{} conjecture states that there is no algorithm preprocessing two lists of $n$ elements
$A,B$ in $n^{2-\Omega(1)}$ time and answering queries ``Does $c$ belong to $A+B$?'' in $n^{1-\Omega(1)}$ time.
Very recently, this conjecture was falsified in two independent papers~\cite{DBLP:journals/corr/abs-1907-11206,DBLP:journals/corr/abs-1907-08355}.

Surprisingly few reductions to 3\SUM{} are known.
It is known that 3\SUM{} is equivalent to convolution 3\SUM{}~\cite{Patrascu10}, which is widely used in the proofs of 3\SUM{}-hardness.
Interestingly, 3\SUM{} that is solved in $\Oh(n^2)$ time is fine-grained equivalent to \textsf{MonoConvolution}, which is solved in $\Oh(n^{3/2})$ time \cite{Lincoln0W20}.
In \textsf{MonoConvolution} we are given three sequences $a,b,c$ and for every index $i$ ask if there is $j$ such that $a_j=b_{i-j}=c_i$.
In addition, Jafargholi and Viola~\cite{JafargholiV16} showed that solving 3\SUM{} in $\tilde{\Oh}(n^{1+\eps})$ time for some $\eps < 1/15$ would lead to a surprising algorithm for triangle listing.

\section{Technical overview}\label{sec:overview}

In this work, we show a series of subquadratic reductions between different generalizations of 3\SUM{}  and \Convolution 3\SUM{}. A subquadratic reduction is formally defined as follows:

\begin{definition}[cf. \cite{WilliamsW18}]\label{def:reduction}
 Let $A$ and $B$ be computational problems with a common size measure $m$ on inputs.
 We say that there is a \emph{subquadratic reduction} from $A$ to $B$ if there is an algorithm $\mathcal{A}$ with oracle access to $B$, such that for every $\eps>0$ there is $\delta>0$ satisfying three properties:
 \begin{enumerate}
  \item For every instance $x$ of $A$, $\mathcal{A}(x)$ solves the problem $A$ on $x$.
  \item $\mathcal{A}$ runs in $\Oh(m^{2-\delta})$ time on instances of size $m$.
  \item\label{item:sizes-analysis} For every instance $x$ of $A$ of size $m$, let $m_i$ be the size of the $i$-th oracle call to $B$ in $\mathcal{A}(x)$.
  Then $\sum_im_i^{2-\eps}\leq m^{2-\delta}$.
 \end{enumerate}
We use the notation $A\red B $ to denote the existence of a subquadratic reduction from $A$ to $B$. If $A\red B$ and $B \red A$, we say that $A$ and $B$ are \emph{subquadratic-equivalent} and denote it $A \eqq B$.
\end{definition}

\paragraph{Formal definitions of 3\LDT{}, 3\SUM{}, and \AVG{}.} 
We work with the following formulations of 1- and 3-partite 3\LDT, where $\bar \alpha$ denotes a triple $(\alpha_1,\alpha_2,\alpha_3)$:

\defproblemwithparameters{1-partite 3\LDT{}$_c(1,\bar \alpha, t)$}
{Set $X$ of $n$ numbers over the universe $[-n^c,n^c]$.}	
{distinct $x_1,x_2,x_3\in X$ such that ${\sum_{i=1}^3 \alpha_i x_i=t}$}
{Integer coefficients $\alpha_1,\alpha_2,\alpha_3$ and $t$, real $c\geq2$}

\defproblemwithparameters{3-partite 3\LDT{}$_c(3,\bar \alpha, t)$}
{Sets $S_1,S_2,S_3$ of $n$ numbers over the universe $[-n^c,n^c]$.}
{$x_1\in S_1,x_2\in S_2, x_3\in S_3$ such that $\sum_{i=1}^3 \alpha_i x_i=t$}
{Integer coefficients $\alpha_1,\alpha_2,\alpha_3$ and $t$, real $c\geq2$}

The 3\SUM{}$_c$ problem is defined as 3\LDT{}$_c(p,(1,1,-1), 0)$, where $p = 1$ or $p = 3$ depending on the partity\footnote{While some works use this definition~\cite{Williams15,Chan18,BaranDP08}, other~\cite{GajentaanO95,AilonC05, EricksonBounds99,GoldS17,GronlundP18} define the $3\SUM$ problem as 3\LDT{}$_c(p,(1,1,1), 0)$. It is well-known that the two variants are subquadratic-equivalent.}. The \AVG{}$_c$ problem introduced by Erickson~\cite{Erickson95} is defined as 3\LDT{}$_c(1,(1,1,-2), 0)$.

\paragraph{Formal definitions of \Convolution3\LDT{}, \Convolution3\SUM{}, and \Convolution\AVG{}.}
Let $n$ be an odd number and $n'=\frac{n-1}{2}$.
We consider the following formulation of \Convolution3\LDT{}:

\defproblemwithparameters{1-partite \Conv$_c(1,\bar\alpha,t)$}{Array $A\ind$ of $n$ numbers over the universe $[-n^c,n^c]$.}{distinct $j_1,j_2,j_3 \in \ind$ such that $\sum_{i=1}^3 \alpha_i j_i=t$ and  $\sum_{i=1}^3 \alpha_i A[j_i]=t$}
{Integer coefficients $\alpha_1,\alpha_2,\alpha_3$ and $t$, real $c\geq1$}
\noindent

\defproblemwithparameters{3-partite \Conv$_c(3,\bar\alpha,t)$}{Arrays $A_1,A_2,A_3\ind$ of $n$ numbers over the universe $[-n^c,n^c]$.}{$j_1,j_2,j_3 \in \ind$ such that $\sum_{i=1}^3 \alpha_i j_i=t$ and  $\sum_{i=1}^3 \alpha_i A_i[j_i]=t$}
{Integer coefficients $\alpha_1,\alpha_2,\alpha_3$ and $t$, real $c\geq1$}

Informally, in \Conv{} we require the indices of elements to satisfy the same condition as their values. Analogously to above, the \Convolution3\SUM$_c$ problem is then defined as \Convolution 3\LDT{}$_c(p,(1,1,-1), 0)$, where $p = 1$ or $p = 3$ denotes whether the variant is 1-partite or 3-partite.  The \Convolution\AVG$_c$ problem, first mentioned by Erickson, can be defined as \Convolution3\LDT{}$_c(1,(1,1,-2), 0)$.
We stress that the arrays consist of positive and negative indices.
At the end of Section~\ref{se:equiv-3ldt-conv4ldt} we describe when and why this is necessary.

Finally, we require that $c\geq2$ for 3\LDT{} and $c\geq1$ for \Conv, as otherwise the instance can be solved in subquadratic time using the fast Fourier transform\footnote{We can represent each array as a set $\{A_i[j]\cdot 4 n + j: j\in \ind\}$, as in Lemma \ref{le:conv_to_nonconv}. Then, we can proceed as in~\cite[Exercise 30.1-7]{CormenLRS09}.}.

\paragraph{Our contribution.}
Consider an instance of 3\LDT{} or \Conv. 
Define the size of an instance to be the number $n$ of elements in the considered sets or arrays.

Note that if any of the coefficients $\alpha_i$ is 0, then we need to find at most two numbers satisfying a linear relation, which can be done in  $\Oh(n\log n)$ time for 3\LDT{} (by first sorting and then for every candidate of $x_3$ scanning the sorted sequence with two pointers) and in $\Oh(n)$ time for \Conv{} (by checking all pairs of indices satisfying the relation). Also, if $t\ne 0$ and $\gcd(\alpha_1,\alpha_2,\alpha_3) \nmid t$ then the instance is obviously a NO-instance, and we can return the answer in constant time. This motivates the following definition:

\begin{definition}
We call a variant of 3\LDT$_c$ or \Conv$_c$ problem with coefficients $\bar\alpha$ and $t$ \emph{trivial}, if either 
\begin{enumerate}
  \item Any of the coefficients $\alpha_i$ is zero, or 
  \item $t\ne 0$ and $\gcd(\alpha_1,\alpha_2,\alpha_3) \nmid t$
\end{enumerate}
and otherwise \emph{non-trivial}.
If the above conditions hold, we call the coefficients $\bar\alpha$ and $t$ \emph{trivial}.
\end{definition}

We show that for all $c\geq 2$ all non-trivial variants of 3\LDT{}$_c$ and \Conv$_{c-1}$ are subquadratic-equivalent. Unless stated otherwise, we establish reductions between two variants of 3\LDT{} (or \Conv) with the same value of parameter $c$, that is over the same universe $U = [-n^c,n^c]$.
To avoid clutter, in the reductions that work for all $c\geq2$ we omit the parameter $c$.
For example, one should read the notation $\ldt1\alpha t\red \ldt3\alpha t$ as: for all $c\geq2$ it holds that $\ldtc1\alpha t\red \ldtc3\alpha t$.
Similarly, for \Conv$_c$ we read the statement $\Conv(1,\bar\alpha,t)\red \Conv(3,\bar\alpha,t)$ as: for all $c\geq1$ it holds that $\Conv_c(1,\bar\alpha,t)
\red \Conv_c(3,\bar\alpha,t)$.
All numbers in the considered problems and reductions are integers. All reductions, unless said otherwise, are deterministic. We assume the standard $w$-bit word RAM model. That is, every word consists of $w$ bits, and standard arithmetic and bitwise
operations can be performed in constant time on such words. Internally, the words are unsigned integers from $[0,2^{w})$, however
by using two's complement we can treat them as signed integers from $[-2^{w-1},2^{w-1})$.
We assume that $w\geq \max\{\log n,\log U\}$.

We first establish equivalence for 3\LDT{}:

\begin{theorem}\label{thm:everything-equivalent}
For all $c\geq2$, all non-trivial variants (1- and 3-partite) of 3\LDT{}$_c$ are subquadratic-equivalent.
\end{theorem}

\noindent In particular, this implies the following.

\begin{corollary}
For all $c\geq2$, \AVG{}$_c$ is subquadratic-equivalent to 3\SUM{}$_c$.
\end{corollary}

\noindent Thus, we completely resolve both Question~\ref{q:avg} and Question~\ref{q:ldt}. In order to design the most interesting of our reductions, from $\ldt 3 \alpha 0$ to $\ldt 1 \alpha 0$, we make use of progression-free sets.
We call a set $S\subseteq \{1,2,\ldots,n\}$ \emph{progression-free}
if it contains no non-trivial arithmetic progression, that is, three distinct elements $a,b,c$ such that $a+b-2c=0$. Erdős and Turan~\cite{ErdosT36}
introduced the question of exhibiting a dense subset with such a property, and presented a construction with $\Omega(n^{\log_{3}2})$
elements. This was improved by Salem and Spencer~\cite{SalemS42} to $n^{1-\Oh(1/\log\log n)}$, and then by
Behrend~\cite{Behrend46} to $\Omega(n / (2^{2\sqrt2\cdot\sqrt{\log n}}\cdot\log^{1/4}n))$. More recently,
Elkin~\cite{Elkin10} showed how to construct a set consisting of $\Omega(n \log^{1/4}n / 2^{2\sqrt2\cdot\sqrt{\log n}})$
elements. One could naturally ask for a dense subset that avoids a certain linear equation $\alpha_{1}x_{1}+\alpha_{2}x_{2}=(\alpha_{1}+\alpha_{2})x_{3}$, where $\alpha_{1},\alpha_{2}$ are positive integers. Indeed, it turns out that Behrend's argument works with minor modifications also for such equations~\cite[Theorem 2.3]{Ruzsa93}. We use an extension of this argument to partition an arbitrary set into a small number of progression-free sets. Average-free sets have been already successfully applied in various areas of theoretical computer science~\cite{ChandraFL83b,CoppersmithW90,HastadW03,WilliamsW13,DellM14,AlonFKS00,AbboudB17,FominGLS14,JansenKMS13}.
In particular, Behrend-like constructions led to conditional lower bounds providing e.g. a reduction from $k$-Clique to $k^2$-SUM \cite{AbboudLW14}, from $k$-SAT to Subset Sum~\cite{AbboudBHS19} and for scheduling problems \cite{DBLP:journals/corr/abs-2003-07113}.

We then extend our techniques to show equivalences for \Conv:

\begin{theorem}\label{th:every-convldt-equivalent}
For all $c\geq1$, all non-trivial variants (1- and 3-partite) of \Conv$_c$ are subquadratic-equivalent. 
\end{theorem}

One can suspect that 3\LDT{} is connected to \Conv{} on a smaller universe as the indices of elements can also convey some information.
We show that this intuition is in fact fully correct.
By adjusting a folklore reduction from \Conv{} to 3\LDT{} \cite{ChanH20} and a modification of \Patrascu's \cite{Patrascu10} and Chan and He's \cite{ChanH20} reduction from 3\SUM{} to \Convolution3\SUM{} we obtain that these problems are equivalent when the size of the universe of considered instances differs by a factor of $n$ between convolution and non-convolution variants.
Finally, by combining this with Theorems~\ref{thm:everything-equivalent} and \ref{th:every-convldt-equivalent} we obtain equivalences between all non-trivial 1- and 3-partite variants of \Conv{} and 3\LDT{}, both convolution and non-convolution:

\begin{restatable}{theorem}{AllLDTEquivalences}\label{th:everything-equivalent}
 For every $c\geq2$, all non-trivial 1- and 3-partite variants of 3\LDT$_c$ and of \Conv$_{c-1}$ are subquadratic-equivalent.
%  $$\mathop{\forall}_{\substack{c\geq2,e_1,e_2\in\{1,3\}\\\text{non-trivial }(\bar\alpha_1,t_1),(\bar\alpha_2,t_2)}	}\text{3\LDT}_c(e_1,\bar\alpha_1,t_1) \eqq \Convolution\text{3\LDT}_{c-1}(e_2,\bar\alpha_2,t_2)$$
\end{restatable}

In the above equivalences we only consider polynomial-size universes. 
Of course one can also consider problems with two parameters: number of elements $n$ and the size of the universe $U$, possibly much bigger than $n$.
However, it turns out that instances of 3\LDT{} over a universe bigger than cubic can be deterministically reduced to instances over the cubic universe.
This follows by extending the result of Fischer et al. \cite{FischerKP24} as explained in detail in Section~\ref{se:universe-reduction}.
Consequently, the only interesting variants of 3\LDT{} are $3\SUM_c$ for $c\in[2,3]$.
Instances over larger universes for any non-trivial variant are equivalent to 3$\LDT_3$, while instances over smaller universes  or of any trivial variant can be solved in subquadratic time.

\begin{restatable}{theorem}{AllLDTAtMost}
 For every $U\geq n^3,p\in\{1,3\}$ and non-trivial coefficients $\bar\alpha,t$, $3\LDT_?(p,\bar\alpha,t)$ over $[-U,U]$ is subquadratic-equivalent to $3\LDT_3(p,\bar\alpha,t)$.
\end{restatable}

By combining this result with reductions between 3\LDT{} and \Conv{}, we can draw a similar conclusion for \Conv{}, that the only interesting non-trivial variants are $\Convolution3\SUM_c$ for $c\in[1,2]$.

\begin{restatable}{theorem}{AllConvLDTAtMost}
 For every $U\geq n^3,p\in\{1,3\}$ and non-trivial coefficients $\bar\alpha,t$, $\Conv_?(p,\bar\alpha,t)$ over $[-U,U]$ is subquadratic-equivalent to $\Conv_2(p,\bar\alpha,t)$.
\end{restatable}
\noindent
We stress that all reductions provided in this paper are deterministic.

\section{Equivalences between different variants of 3\LDT}
\label{se:3LDT}
In this section, we show Theorem~\ref{thm:everything-equivalent} through a series of reductions depicted in Figure~\ref{fig:ldtreductions}.
In order to show a subquadratic reduction $A\red B$ we often present only a reduction from an instance of~$A$ to a number of instances of $B$.
If the analysis of the sizes of the obtained instances (whether they satisfy Property \ref{item:sizes-analysis} of Definition~\ref{def:reduction} or not) is immediate, we omit it.
Here and below we use the following notation: when writing $\sum_i$ we mean the sum over all possible values of $i$; $[k]=\{1,2,\ldots,k\}$; $f[A]=\{f(a):a\in A\}$ is the image of $f$ over $A$;  sumset $A+B$ is defined as $\{a+b:a\in A,b\in B\}$ and in particular we can add an element to a set: $A+x=A+\{x\}=\{a+x:a\in A\}$.

\begin{figure}[h]
    \centering
    \begin{tikzpicture}
    \begin{scope}[every node/.style={align=center},
                  every edge/.style={draw=black}]

	\node (n21) at (4.5,0) {$\ldt 3 {\alpha_2} {t_2}$};
    \node (n23) at (4.5,1.5) {$\ldt 3 {\alpha_1} {t_1}$};
    \node (n31) at (9,0) {$\ldt 1 {\alpha_2} {t_2}$};
    \node (n33) at (9,1.5) {$\ldt 1 {\alpha_1} {t_1}$};
        
        \path [->] (n33) edge node[below] {\tiny{Lem. \ref{lm:from_1_to_3}}} (n23);
        \path [<->] (n21) edge node [left] {\tiny{Lem.~\ref{le:3-partite-equivalent}}} (n23);
        \path [->] (n21) edge node [below] {\tiny{Lem.~\ref{le:from_3_to_1_nonzero_t}~\&~\ref{le:from_3_to_1_zero_t}}} (n31);
    \end{scope}
    \end{tikzpicture}
    \caption{Subquadratic reductions between different variants of 3\LDT. Subscripts to coefficients $\alpha$ and $t$ denote if the reduction allows changing the coefficients or not. 
    In all the reductions the value of parameter $c$ is preserved and hence not shown.
}
    \label{fig:ldtreductions}
\end{figure}

We start by showing technical lemmas that allow us to handle instances of 3\LDT{} in which sets have at most $n$ elements and their elements can exceed the $[-n^c,n^c]$ range, but are within the range $[-dn^c,dn^c]$ for some constant $d\geq 1$.
We call such instances \emph{semi-instances}.
In such a case we add a number of elements that will never be a part of triple $\bar x$ of elements satisfying $\sum_i\alpha_ix_i=t$.

\begin{proposition}\label{prop:finding_betas}
For any triple $\bar\alpha$ of nonzero coefficients we can find in constant time integer coefficients $\bar\beta$ such that $\forall_{\emptyset \ne S \subseteq [3]} \sum_{i\in S} \alpha_i \beta_i \ne 0$.
 \end{proposition}
 \begin{proof}
 We need to find coefficients $\bar\beta$ satisfying all seven conditions of the form $\sum_{i\in S} \alpha_i \beta_i \ne 0$ for $\emptyset \ne S \subseteq [3]$.
  First we set $\beta_1=1$.
  As $\alpha_i\ne0 $, we have $\alpha_1\beta_1\ne 0$.
  Next, by setting $\beta_2=\max\{0,\lceil\frac{\alpha_1\beta_1}{-\alpha_2}\rceil\}+1$ we guarantee that $\alpha_2\beta_2\ne 0$ and $\alpha_1\beta_1+\alpha_2\beta_2\ne 0$.
  Finally, we set $\beta_3=\max\{0,\lceil\frac{\alpha_1\beta_1}{-\alpha_3}\rceil, \lceil\frac{\alpha_2\beta_2}{-\alpha_3}\rceil, \lceil\frac{\alpha_1\beta_1+\alpha_2\beta_2}{-\alpha_3}\rceil  \}+1$ and fulfill all the conditions with $3 \in S$.
  Hence the above choice of $\bar\beta$ fulfills the required seven conditions.
 \end{proof}

 \begin{lemma}\label{le:handle_semi_instance}
  Let $\bar\alpha$ be non-zero coefficients, $c\geq2,d\geq1$ and $t$ be arbitrary.
  Given a semi-instance of $\ldtc3\alpha t$ with sets $S_1,S_2,S_3$ of at most $n$ elements from $[-n^c,n^c]$, we can construct in $\Oh(n)$ time an equivalent instance of $\ldtc3\alpha t$, with sets $S_1',S_2',S_3'$ of $n'=\Oh(n)$ elements from $[-\frac1d(n')^c,\frac1d(n')^c]$ such that there exists a triple $\bar x\in S_1\times S_2\times S_3$ of elements such that $\sum_i\alpha_ix_i=t$ iff there exists a triple $\bar x'\in S_1'\times S_2'\times S_3'$ such that $\sum_i\alpha_i x_i'=t$.
 \end{lemma}
 \begin{proof}
  By Proposition~\ref{prop:finding_betas} we can find $\bar \beta$ such that $\forall_{\emptyset \ne S \subseteq [3]} \sum_{i\in S} \alpha_i \beta_i \ne 0$ and let $\beta^*=\max_i|\beta_i|$.
  Let $M=2n^c\cdot\sum|\alpha_i|$ and $n'$ satisfy $(n')^c\geq2\beta^*Md$, so $n'=\lceil n(4d\beta^*\sum_i|\alpha_i|)^{1/c}\rceil$.
  Clearly $n' \geq n$ and for sufficiently big $n$ we have $n'\leq n^c$.
  We set $S_i'=S_i\cup \{\beta_i M + j: j\in\{1,2,\ldots,n'-|S_i|\}\}$.
  Now we show that the sets $S_i'$ satisfy the desired properties.
  Clearly, the above construction runs in $\Oh(n)$ time and sets $S_i'$ have $n'=\Oh(n)$ elements each.
  Next, $S_i'\subseteq [-\frac1d(n')^c,\frac1d(n')^c]$ because $\frac1d(n')^c\geq 2\beta^*M\geq n^c$, so $S_i \subseteq [-n^c,n^c] \subseteq [-\frac1d(n')^c,\frac1d(n')^c]$ and $S_i'\setminus S_i \subseteq [\beta_iM +1,\beta_iM_n']\subseteq [-2\beta^*M,2\beta^*M]$.
  If there is a triple $\bar x$ such that $x_i\in S_i$ and $\sum \alpha_ix_i = t$ then it is also a solution for $S_1',S_2',S_3'$.

  Now we show the opposite direction.
  Suppose there is a triple $\bar x'$ such that $x_i'\in S_i'$ and $\sum \alpha_ix_i' = t$.
  Define $\bar \psi$ as $\psi_i=\beta_i$ if $x_i'\in S_i'\setminus S_i$ and $\psi_i=0$ if $x_i'\in S_i$.
  Then $|\sum\alpha_ix_i'-\sum\alpha_i\psi_iM| \leq \sum|\alpha_i|n^c = M/2$.
  By Proposition~\ref{prop:finding_betas} $\sum\alpha_i\psi_iM =0$ iff $\psi_1=\psi_2=\psi_3=0$.
  Otherwise $|\sum\alpha_i\psi_iM|\geq M$ so $|\sum \alpha_ix_i|\geq M/2$ which contradicts the assumption that $\sum \alpha_ix_i=t$.
  Hence for every triple $\bar x'$ such that $x_i'\in S_i'$ and $\sum \alpha_ix_i=t$ we have that $x_i'\in S_i$ which concludes the proof.
 \end{proof}

Now we show equivalence of all non-trivial 3-partite variants of 3\LDT{} and a reduction from 1-partite variant to 3-partite variant.
To avoid clutter, till the end of this section we do not write the parameter $c$ in the description of the considered variants as $c$ is preserved in all the presented reductions and satisfies $c\geq 2$.

\begin{lemma}\label{le:3-partite-equivalent}
All non-trivial 3-partite variants of 3\LDT{} are subquadratic-equivalent.
\end{lemma}
\begin{proof}
We need to show a reduction between any two non-trivial 3-partite variants of 3\LDT{}.
To this end, we establish three reductions: $\ldt 3\alpha0\red\ldt 3\alpha t$ and $\ldt 3\alpha t\red\ldt 3\alpha 0$  and
finally $\ldt 3\alpha0\red\ldt 3\beta0$ for $\bar\alpha,\bar\beta$ and $t\ne 0$ such that all the considered variants are non-trivial.
In each of those reductions we need to make sure that we obtain an instance of 3\LDT{} and not a semi-instance.
Each of them requires the same additional clean-up stage that we present at the end of this proof.
Reductions between other variants can be obtained by combining at most three of the above.

\begin{enumerate}
\item $\ldt 3\alpha0\red\ldt 3\alpha t$. We have
$\gcd(\alpha_1,\alpha_2,\alpha_3) | t$ because $\bar\alpha$ and $t$ are non-trivial coefficients, so by the Chinese remainder theorem there exists an integer triple $\bar y$ such that $\sum_i \alpha_i y_i=t$.
Given the three sets $S_1,S_2,S_3$ we construct three sets $S_1',S_2',S_3'$ where $S'_i=\{x+y_i:x\in S_i\}$.
Then there is $\bar x\in S_1\times S_2\times S_3$ satisfying $\sum_i \alpha_i x_i=0$ iff there is $\bar x'\in S'_1\times S'_2\times S'_3$ satisfying $\sum_i \alpha_i x'_i=t$.
 \item $\ldt 3\alpha t\red\ldt 3\alpha 0$. As above but by subtracting the $y_i$ terms.
 \item $\ldt 3\alpha0\red\ldt 3\beta0$. Define $q=\lcm(\beta_1,\beta_2,\beta_3)$ so that $\frac{\alpha_{i} q}{\beta_{i}}$ is an integer.
Given $S_1,S_2$ and $S_3$ we construct $S_1',S_2'$ and $S_3'$ by setting $S'_i=\{x\frac{\alpha_i q}{\beta_i}:x\in S_i\}$.
\end{enumerate}
Recall that all the considered variants of 3\LDT{} are over the same universe $[-n^c,n^c]$ for some $c\geq2$.
Now we need to handle the situation that sets $S_1',S_2',S_3'$ form a semi-instance of 3\LDT, because after the above transformation the elements can be outside the range $[-n^c,n^c]$.
In the first two cases, the universe increases at most by an additive factor $\max|y_i|=\Oh(1)$, so not more than by  multiplicative factor $2$.
In the third case the universe increases at most by factor $q\cdot\max|\frac{\alpha_i}{\beta_i}|$.
Let $w=\max\{2,q\cdot\max_i|\frac{\alpha_i}{\beta_i}|\}$.
Before applying the reduction, we first apply Lemma~\ref{le:handle_semi_instance} with $d=w$ obtaining sets of $n'=\Theta(n)$ elements from  $[-\frac1d(n')^c,\frac1d(n')^c]$.
Then after multiplying or adding a constant to an element, it still belongs to $[(n')^c,(n')^c]$.
\end{proof}

To show $\ldt1\alpha t\red \ldt3\alpha t$, we apply the folklore reduction from $1$-partite 3$\SUM$ to $3$-partite 3$\SUM$ based on the color-coding technique of Alon et al.~\cite{AlonYZ95}.
For completeness we present the proof below.

\begin{lemma}\label{lm:from_1_to_3}
 For all $\bar \alpha$ and $t$, $\ldt1\alpha t\red \ldt3\alpha t$.
\end{lemma}
\begin{proof}
In this reduction, given one set $X$ we need to create a number of 3-partite instances of 3\LDT{} in such a way that there exist distinct $x_{1},x_{2},x_{3}\in X$ 
satisfying the given equation iff at least one of the 3-partite instances is a YES-instance.
The reduction will not change coefficients $\alpha_i$ and the parameter $t$.
Note that simply creating a single 3-partite instance by making all three sets equal to $X$ does not work,
as we are not able to forbid taking the same element of $X$ more than once.

We use the color-coding technique that was introduced by Alon et al.~\cite{AlonYZ95}, in which we choose a number of colorings of the elements
of $X$ with $k$ colors in such a way that, for every $k$-element subset of $X$, there is a coloring in which all elements from the subset have
distinct colors.
This can be achieved with high probability by simply choosing sufficiently many random colorings, but we will use the deterministic construction by Schmidt and Siegel~\cite{SchmidtS90}.

\begin{fact}[cf.~\cite{SchmidtS90}]\label{fact:color-coding}
There exists a family $F$ of $2^{\Oh(k)}\log^{2} n$ functions $[n]\rightarrow [k]$ such that, for every $k$-element set $Y\subseteq [n]$, there exists a function $f\in F$ with $|f[Y]|=k$. Each function is described by a bit string of length $\Oh(k)+2\log\log n$ and, given
constant-time read-only random access to the bit string describing $f\in F$ and any $x\in[n]$, we can compute $f(x)$ in constant time.
\end{fact}

We work with $k=3$, so the above fact gives us a family $F$ consisting of $\Oh(\log^{2} n)$ functions.
Given a set $X=\{x_1,x_2,\ldots, x_n\}$, for every function $f\in F$ and every permutation $\pi \in S_{3}$ we obtain a 3-partite semi-instance of 3\LDT{} by setting $S_{\pi(i)} = \{x_c: f(c)=i\}$ for $i\in [3]$.
In every 3-partite semi-instance the sets $S_{i}$ correspond to a partition of the original set $X$, and
for any distinct $x_{1},x_{2},x_{3}\in X$ there exists a 3-partite semi-instance such that $x_{1}\in S_{1}$, $x_{2}\in S_{2}$, and $x_{3}\in S_{3}$.
Thus, we showed how to reduce a 1-partite instance of 3\LDT{} to $\Oh(\log^{2} n)$ semi-instances of 3-partite 3\LDT{} with the same coefficients $\bar\alpha$ and $t$. The reduction works in $\Oh(n\log^{2}n)$ time.
Finally, we reduce every semi-instance to an instance of 3-partite 3\LDT{} by Lemma~\ref{le:handle_semi_instance} with $d=1$.
\end{proof}

It remains to show how to reduce an arbitrary non-trivial 3-partite variant of 3\LDT{} to a 1-partite one with the same coefficients $\bar\alpha$ and $t$. Before proceeding to the reduction, we show a few preliminary lemmas. Let~$C$ be a sufficiently big constant to be fixed later. Recall that in definition of $\ldt3\alpha t$ we have three sets $S_i$, and in the definition of $ \ldt1\alpha t$ one set $X$. 
We would like to construct the set $X$ by setting $X=\bigcup_i \{Cx+\gamma_i : x\in S_i\}$, where $\bar\gamma$ are pairwise distinct
coefficients chosen so as to ensure that all triples $\bar x$ consisting of distinct elements from $X$ satisfying $\sum_i \alpha_ix_i=t$ also satisfy that
$x_i$ corresponds to an element of $S_{i}$, for every $i\in[3]$.
For example, for 3\SUM{} we can set $X= \{3C+x: x\in S_1\} \cup \{C+x: x\in S_2\} \cup \{4C+x: x\in S_3\}$.
Now we extend this approach to arbitrary coefficients.

For a triple $\bar x$ of elements from~$X$, an~\textit{origin} is a function $f:[3]\rightarrow [3]$ such that, for every $i\in[3]$, $x_{i}$ corresponds to an element of $S_{f(i)}$.
Clearly $f$ is a function, because $\gamma$'s are pairwise distinct.
For instance, consider $\bar x = (Cx+\gamma_2,Cy+\gamma_3, Cz+\gamma_3)$ where $x\in S_2$ and $y,z\in S_3$.
Then we have $f(1)=2$ and $f(2)=f(3)=3$.
Intuitively, now we would like to find coefficients $\bar \gamma$ such that for every triple $\bar x$ of elements from $X$ such that $\sum_i \alpha_ix_i = t$ their origin is the identity function ($f(i)=i$ for $i\in [3]$), so in particular we need to forbid using more than one number from the same set ($|\{f(i) : i\in[3] \}|<3$).
However, some coefficients from $\bar\alpha$ might be equal, so we also need to allow  origins in which we permute the elements with the same values of $\alpha_i$.
For example, if all coefficients $\alpha$ are equal, we allow every origin that is a permutation of $[3]$, and when $\alpha_1=\alpha_3\ne\alpha_2$, we have two allowed origins: $(1,2,3)$ and $(3,2,1)$.
This is formalized in the following definition.

\begin{definition}
For any coefficients $\bar\alpha$, we call an origin $f$ \emph{allowed} if
$\forall {i\in [3]} \{f(x):x\in[3],\alpha_x=\alpha_i\} = \{x: x\in[3],\alpha_x =\alpha_i\}$, and otherwise we call it \emph{forbidden}. In addition, if $f(1)=f(2)=f(3)$ we call the origin \emph{constant}.
\end{definition}

We show that it is always possible to find a triple $\bar\gamma$ which excludes solutions from most of the forbidden origins.
In other words, we present how to find $\bar\gamma$ such that for every triple $\bar x$ of elements from the constructed set $X$ such that $\sum_i \alpha_ix_i = t$, we have that origin of $\bar x$ is either allowed or constant.
Additionally, we need to ensure that in the allowed origin the summands not multiplied by $C$ cancel out, so we require that $\sum_i \alpha_{i}\gamma_i = 0$.

\begin{lemma}\label{le:finding-gammas}
 For any triple $\bar\alpha$ of nonzero coefficients there exists a triple $\bar\gamma$ of nonzero, pairwise distinct coefficients such that $\sum_i\alpha_i\gamma_i=0$ and for every non-constant forbidden origin~$f$ we have $\sum_i\alpha_i\gamma_{f(i)}\ne 0$.
\end{lemma}
\begin{proof}
Consider the 3-dimensional space $\mathbb{Q}^{3}$.
Clearly, the set of all triples $\bar\gamma$ such that $\sum_i\alpha_i\gamma_i=0$ spans a plane there, we denote it $\Gamma_{id}$.
There are less than $3^3=\Oh(1)$ non-constant forbidden origins $f$ and each of them corresponds to an equation
$\sum_i\alpha_i\gamma_{f(i)}=0$ that must be avoided, which also corresponds to a forbidden plane $\Gamma_f$.
By the definition of a forbidden origin $f$, we have $\Gamma_f\ne \Gamma_{id}$.
Moreover, even if $\sum_i\alpha_i=0$, $\Gamma_f$ is not the whole space $\mathbb{Q	}^{3}$, as $f$ is a non-constant configuration.
Next, as we need all the coefficients $\gamma_i$ to be nonzero, we add forbidden planes $\Gamma_i=\{\bar\gamma: \gamma_i=0\}$, for $i\in [3]$.
Similarly, as we need all the coefficients $\gamma_i$ to be pairwise distinct, we add forbidden planes
$\Gamma'_i=\{\bar\gamma: \gamma_i=\gamma_{(i+1)\bmod 3+1}\}$, for $i\in [3]$.
Clearly, $\Gamma_i\ne\Gamma_{id}$ and $\Gamma_i'\ne\Gamma_{id}$ because the coefficients $\bar\alpha$ are nonzero.
Then let $\mathcal{F}=\{\Gamma_f:f\text{ is non-constant and forbidden}\}\cup \{\Gamma_i:i\in[3]\}\cup \{\Gamma_i':i\in[3]\}$ be the set of all forbidden planes.
Now we need to show that $\Gamma_{id}\setminus\bigcup_{F\in\mathcal{F}}F\ne \emptyset$, using the assumption that $\forall_{F\in\mathcal{F}} F\ne\Gamma_{id}$.

Clearly both $\Gamma_{id}$ and all planes $f\in\mathcal{F}$ contain the origin $o=(0,0,0)$.
Consider an arbitrary line $\ell \subset \Gamma_{id}$ that does not pass through the origin $o$ and contains infinitely many points with all coordinates rational.
For example, we can take the line passing through $(1,0,-\alpha_{1}/\alpha_{3})$ and $(0,1,-\alpha_{2}/\alpha_{3})$. 
Observe that for any $F\in\mathcal{F}$, if $|\ell \cap F| \geq 2$ there would be three non-collinear points (two
from $\ell$ and $o$) belonging to two distinct planes $\Gamma_{id}$ and $F$, so contradiction.
Hence $\ell\cap F$ is either empty or a point.
Recall that there is a constant number of planes in $\mathcal{F}$.
Then $\Gamma_{id}\setminus\bigcup_{F\in\mathcal{F}}F\supseteq \ell\setminus\bigcup_{F\in\mathcal{F}}F$ contains
some point with rational coordinates, because there are infinitely many such points on $\ell$.
This gives us a point in $\mathbb{Q}^{3}$ that belongs to $\Gamma_{id}$ and does not belong to any $F \in \mathcal{F}$. By scaling its coordinates to integers, we obtain~$\bar\gamma$.
\end{proof}
\noindent
Now we are ready to show the reduction from $\ldt3 \alpha t$ to $\ldt1 \alpha t$ for $t \ne 0$.

%For any $\bar\alpha$, every 3-partite instance of $\ldt3\alpha 0$ on at most $n$ numbers can be subquadratically reduced to a constant number of instances of 1-partite $\ldt1\alpha 0$ on at most $n'=\Oh(n)$ numbers.

\begin{lemma}\label{le:from_3_to_1_nonzero_t}
Assume $t$ and $\bar\alpha$ are non-zero.
For any $\bar\alpha$, every 3-partite instance of $\ldt3\alpha t$ reduces in linear time to a 1-partite instance of $\ldt1\alpha t$.
\end{lemma}
\begin{proof}
If the considered variant of 3\LDT{} is trivial, we can solve it in $\Oh(n)$ time and terminate.
Otherwise, by the extended Euclidean algorithm, we can choose an integer triple $\bar y$ such that $\sum_i \alpha_iy_i=t$ and let $\mathcal{Y}(\bar\alpha,t)=\max_i|y_i|$.
We apply Lemma~\ref{le:finding-gammas} on $\bar\alpha$ to obtain $\bar\gamma$ and construct the set $X$ as follows:

$$X=\bigcup_i \{C^2(x-y_i) + C\gamma_i + y_i: x \in S_i\},$$
where $C$ is a sufficiently big constant such that the absolute value of any linear combination of $\gamma$'s or $y$'s with coefficients $\alpha_{i}$ is smaller than $C$ (for example, we can take $C=1+\left(\max_i\max\{|\gamma_i|,|y_i|\}\right)\cdot\sum_i |\alpha_i|$).
If there is a triple $\bar x$ such that $x_i \in S_i$ and $\sum_i \alpha_i x_i = t$, then by the choice of $\bar \gamma$ and $\bar y$ we have $\sum_i \alpha_i z_i = t$, where $z_i = C^2(x_i-y_i) + C\gamma_i + y_i$.
As coefficients $\gamma_i$ are pairwise distinct, elements $z_i$ are pairwise distinct as well.
Recall that, for a given element $z\in X$, the function $f$ returns the index $i$ such that $z$ comes from an element of $S_i$.
Now consider a triple $\bar z$ such that $z_i \in X$ and $\sum \alpha_i z_i = t$. Let $z_i = C^2 (x_{f(i)}-y_{f(i)}) + C\gamma_{f(i)}  + y_{f(i)}$, where $x_{f(i)} \in S_{f(i)}$. By the definition of $C$ and the fact that $\sum_i\alpha_iz_i=t$, it holds that $\sum\alpha_i (x_{f(i)}-y_{f(i)})=0$, $\sum \alpha_i \gamma_{f(i)}=0$ and $\sum \alpha_i y_{f(i)}=t$.  We will show that $f$ is an allowed origin which guarantees that $x_{f(1)},x_{f(2)},x_{f(3)}$ is a valid solution of $\ldt3\alpha t$.

By Lemma~\ref{le:finding-gammas}, $\sum_i \alpha_i\gamma_{f(i)} = 0$ implies that the origin $f$ is either constant or allowed.
If $f$ is constant, there exists $j\in[3]$ such that $f(i) = j$ for all $i \in [3]$, so from the fact that $\sum_i \alpha_i \gamma_{f(i)}=0$ we have $\sum_i \alpha_i \gamma_j = 0$ and therefore $\sum_i \alpha_i = 0$ as $\gamma_j \neq 0$. It implies $\sum_i \alpha_i y_{f(i)} = \sum_i \alpha_i y_{j} = 0 \neq t$, hence $f$ cannot be constant and is allowed.

Recall that the considered numbers from sets $S_i$ are from the universe $[-U,U]$ where $U=n^c$ and $c$ is a parameter of the considered variant.
For $n$ large enough, $U\geq\mathcal{Y}(\bar\alpha,t)=\max_i|y_i|$. Consequently, $|x-y_i|\leq 2U$, so the absolute value of the elements of $X$ is at most $2C^2U+C^2+C\leq 3C^2U$, so we obtain a semi-instance of $\ldt1\alpha t$.
To avoid that, we first apply Lemma~\ref{le:handle_semi_instance} on sets $S_1,S_2,S_3$ with $d=3C^2$ and then create set $X$ from the obtained sets $S_1',S_2',S_3'$, which concludes the reduction to an instance of $\ldt1\alpha t$.
\end{proof}

Surprisingly, the case $t = 0$ is more difficult. We would like to proceed as in Lemma~\ref{le:from_3_to_1_nonzero_t}, which is enough to exclude all non-constant forbidden origins and, if $\sum_{i}\alpha_{i} \neq 0$, also the constant origins. However, if $\sum_i\alpha_i=0$, then we cannot exclude the constant origins. In other words, no matter what the chosen $\gamma$'s are we are not able to exclude the solutions that use three distinct elements corresponding to the elements of the same set $S_{j}$. This suggests that we should partition each of the sets $S_j$ into a few sets that contain no triple $\bar x$ of distinct elements such that $\sum_i\alpha_ix_i=0$. To this end, we introduce the following definition:

\begin{definition}
 For any $\gamma,\delta>0$, a set $X$ is \emph{$(\gamma,\delta)$-free} if no three distinct elements $a,b,c\in X$ satisfy $\gamma a+\delta b=(\gamma+\delta)c$.
\end{definition}
\noindent
Now we show that we can always partition an arbitrary subset of $[N]$ into $2^{\Oh(\sqrt{\log N})}$ $(\gamma,\delta)$-free sets\footnote{The idea of this proof is borrowed from \cite{JinX23}, who suggested how to improve the earlier version of the proof that appeared in the conference version of this work \cite{DudekGS20}.
Before we used probabilistic argument to obtain the partition from an arbitrary construction of dense $(\gamma,\delta)$-free sets (possibly more dense than the Behrend's one) that provides only oracle membership access, but in fact this is not necessary in our reduction from 3-partite to 1-partite instances of 3\LDT.}.

\begin{theorem}[cf.~\cite{Behrend46,Ruzsa93}]\label{thm:partition_with_Behrends}
For any $\gamma,\delta>0$ and a set $X \subseteq[N]$, it is possible to construct $e=~2^{\Oh(\sqrt{\log N})}$ sets $X_1, X_2, \ldots, X_e$
such that every $X_j$ is $(\gamma,\delta)$-free and $\bigcup_j X_j = X$.
The construction is deterministic and runs in $\Oh(|X|)$ time.
\end{theorem}
\begin{proof}
 Let $p=2(\gamma+\delta)+1,q=2^{\sqrt{\log N}}$ and $r=\lceil \frac qp \rceil$.
 We represent every element $x\in X$ in base~$q$: $x=\sum_{i=0}^{d-1}x_iq^i$ where $d=\lceil \log_q N\rceil=\lceil \sqrt{\log N} \rceil$ and $x_i\in\{0,\ldots,q-1\}$.
 For every $i$, define $\tilde{x}_i=\lfloor x_i/r \rfloor \in \{0,\ldots p-1\}$ and $x'_i=x_i \bmod r\in \{0,\ldots,r-1\}$. We put an element $x$  to a set $X_\tau$, where the index $\tau$ is a tuple $(\tilde{x}_0,\ldots,\tilde{x}_{d-1};||x'||_2^2)$.
 As $||x||_2^2< dr^2$, there will be at most $e=p^ddr^2\leq p^{d}dq^2=2^{\Oh(\sqrt{\log N})}$ distinct sets.
 
 We claim that the sets $X_\tau$ are $(\gamma,\delta)$-free.
 Suppose the contrary, that there are three distinct elements $a,b,c\in X_\tau$ such that $\gamma a+\delta b = (\gamma+\delta)c$.
 By combining that with representations of $a,b,c$ in base $q$ we obtain:
 
 $$\gamma\cdot \sum_i q^i(a_i'+r\cdot\tilde a_i)+\delta\cdot\sum_i q^i(b_i'+r\cdot\tilde b_i)= (\gamma+\delta)\cdot \sum_i q^i(c_i'+r\cdot\tilde c_i)$$
 By the definition of $X_\tau$ we have $\tilde a_i = \tilde b_i = \tilde c_i$ for $0\leq i < d$ so the above equation simplifies to
 $$\sum_i q^i(\gamma\cdot a_i'+\delta\cdot b_i'- (\gamma+\delta)\cdot c_i')=0$$
 Recall that $x_i'\leq r-1 < \frac{q}{2(\gamma+\delta)+1}$, so there is no carry between different base-$q$ digits of left-hand side of the expression, and for every $0\leq i<d$ we have $\gamma a_i'+\delta b_i' = (\gamma+\delta) c_i'$.
 Considering vectors $a',b',c'$, this gives us $\gamma a' + \delta b' = (\gamma+\delta) c'$.
 We combine it with the triangle inequality and the fact that $||a'||_2=||b'||_2 =||c'||_2$ from the definition of $X_\tau$:
 $$(\gamma+\delta)||c'||_2=\norm{(\gamma+\delta)c'}_{2}=\norm{\gamma a'+\delta b'}_{2} \leq \norm{\gamma a'}_{2}+\norm{\delta b'}_{2}=(\gamma+\delta)||c'||_2$$
which becomes an equality if and only if $a'$ and $b'$ are collinear.
Their norms are equal because $a, b \in X_\tau$, so we finally obtain $a'= b' = c'$ which contradicts the distinctness of $a,b$ and $c$. 
 \end{proof}

\begin{corollary}\label{cor:Behrend_of_at_least_s_elements}
 For any $\gamma,\delta,s>0$, in $s\cdot2^{\Oh(\sqrt{\log s})}$-time we can construct a $(\gamma,\delta)$-free set of $s$ elements from $[s\cdot2^{\Oh(\sqrt{\log s})}]$.
\end{corollary}
\begin{proof}
 Recall that $[n]=\{1,\ldots,n\}$.
 Let $d$ satisfy that the value  $e=~2^{\Oh(\sqrt{\log N})}$ from Theorem~\ref{thm:partition_with_Behrends} is bounded by  $e\leq 2^{d\sqrt{\log N}}$ and let $X=[s\cdot2^{d\sqrt{2\log s}}]$.
 By Theorem~\ref{thm:partition_with_Behrends} for $N=s^2$, we can partition~$X$ into at most $2^{d\sqrt{2\log s}}$ sets, so at least one of them contains at least $s$ elements.
\end{proof}

\begin{lemma}\label{le:from_3_to_1_zero_t}
For every $\bar\alpha$ it holds $\ldt3\alpha 0 \red \ldt1\alpha 0$.
\end{lemma}
\begin{proof}
We show that every 3-partite instance of $\ldt3\alpha 0$ can be reduced to a number of instances of 1-partite $\ldt1\alpha 0$.
If the considered variant of 3\LDT{} is trivial, we can solve it in $\Oh(n)$ time and terminate.
If $\sum_{i}\alpha_{i} \neq 0$, we apply Lemma~\ref{le:finding-gammas} on $\bar\alpha$ to obtain $\bar\gamma$ and construct the set $X=\bigcup_i\{Cx+\gamma_i:x\in S_i\}$ where $C$ is a sufficiently big constant chosen as in Lemma~\ref{le:from_3_to_1_nonzero_t}.
All the coefficients $\gamma_i$ are pairwise distinct, so there is no element added to $X$ more than once.
Similarly as in Lemma~\ref{le:from_3_to_1_nonzero_t}, this definition of set $X$ is enough to exclude all non-constant forbidden origins and, if $\sum_{i}\alpha_{i} \neq 0$, also the constant origins.

Consider now the case $\sum_i\alpha_i=0$. Without loss of generality, there is a permutation $i_{1},i_{2},i_{3}$ of~$[3]$ such that $\alpha_{i_1},\alpha_{i_2}$ are positive and $\alpha_{i_3}$ is negative. (If this is not the case, we first multiply all the coefficients by $-1$, apply our reduction to 1-partite variant and in the end we again multiply the coefficients of the obtained instances by $-1$.)
We apply Theorem~\ref{thm:partition_with_Behrends} to partition every set~$S_{i}$ into $e = 2^{\Oh(\sqrt{\log N})}$ $(\alpha_{i_1}, \alpha_{i_2})$-free subsets $S_{i,j}$, where $i\in[3],j\in [e]$. (We shift $S_i$ by $U$ to apply the theorem, and then shift the resulting sets $S_{i,j}$ by $-U$.) We reduce the  instance of $\ldt3\alpha0$ to~$e^3$ semi-instances of $\ldt1\alpha0$ by considering all possible combinations of the subsets and applying the construction for the case $\sum_{i}\alpha_{i} \neq 0$, which excludes all non-constant forbidden origins by the choice of~$\bar \gamma$ and all constant origins by the partition. To show that the reduction is subquadratic, we must analyze the sizes of the obtained instances. As $2^{\Oh(\sqrt{\log n})}<n^{\eps'}$ for any~$\eps'>0$, we have $2^{\Oh(\sqrt{\log n})}\cdot n^{2-\eps} < n^{2-\delta}$ for all $0<\delta<\eps$.

Finally, we note that we obtain semi-instances of $\ldt1\alpha 0$, because after the partition we have less than $n$ elements and the constructed set $X\subseteq[-2Cn^c,2Cn^c]$.
Similarly as in Lemma~\ref{le:from_3_to_1_nonzero_t}, before constructing $X$ we apply Lemma~\ref{le:handle_semi_instance} with $d=2C$, which gives us instances of $\ldt1\alpha 0$.
\end{proof}

\section{Equivalences between different variants of \Convolution3\LDT}
In this section, we show Theorem~\ref{th:every-convldt-equivalent} that claims that for all $c\geq1$ all non-trivial variants of \Convolution3\LDT{}$_c$ are subquadratic-equivalent. A scheme of reductions is shown in Figure~\ref{fig:conv-only-reductions}. 
To avoid clutter, till the end of this section we do not write the parameter $c$ in the description of the considered variants as $c$ is preserved in all the presented reductions and satisfies $c\geq 1$.

\begin{figure}[h]
    \centering
    \begin{tikzpicture}
    \begin{scope}[every node/.style={align=center},
                  every edge/.style={draw=black}]

	\node (n21) at (0,0) {$\Convolution\ldt 3 {\alpha_1} {t_1}$};
    \node (n23) at (5,0) {$\Convolution\ldt 3 {\alpha_2} {t_2}$};
    \node (n31) at (10,0) {$\Convolution\ldt 1 {\alpha_2} {t_2}$};
        
        \path [<->] (n21) edge node [below] {\tiny{Lem.~\ref{le:3-partite-conv-equivalent}}} (n23);
        \path [<->] (n23) edge node [below] {\tiny{Thm.~\ref{thm:conv_1_equiv_3}}} (n31);
    \end{scope}
    \end{tikzpicture}
    \caption{Subquadratic reductions between different variants of \Convolution3\LDT.
    Subscripts to coefficients $\alpha$ and $t$ denote if the reduction allows changing the coefficients or not. 
    All the reductions preserve the parameter $c\geq 1$ describing the size of the universe.
}
    \label{fig:conv-only-reductions}
\end{figure}

The arguments are similar to those that we used in Section~\ref{se:3LDT}, but we provide them all for completeness.
A subtle difference is that for arrays we cannot take a subset of elements (which was possible for sets) and we need to replace the elements to be removed with some dummy values that can never be part of a solution.
Similarly as in Lemma~\ref{le:handle_semi_instance}, we show that we can remove a subset of elements from the arrays by replacing them with some dummy values that can never be part of a solution to \Conv.
Additionally, we take into account possible increase of the size of the universe and extension of the arrays, measured by $d$ and $e$ respectively.
 
\begin{lemma}\label{le:perp}
 Let $\bar\alpha$ be non-zero, $c,d,e\geq 1$, $t$ be arbitrary constants.
 For every $n$, we can calculate in constant time values $n''=\Oh(n)$ and $\bar\perp$ such that $\perp_1,\perp_2,\perp_3\in[-(n'')^c,(n'')^c]$, $dn^c \leq (n'')^c$ and $en \leq n''$ and for all $\emptyset\ne S \subseteq [3]$ and $\bar z \in [-dn^c,dn^c]^3$ we have:
 $\sum_{i\in S} \alpha_i\perp_i + \sum_{i\in [3]\setminus S} \alpha_i z_i \ne t$.
\end{lemma}
\begin{proof}
By Proposition~\ref{prop:finding_betas}, for given non-zero coefficients $\bar\alpha$ we can find coefficients $\bar\beta$ such that $\forall_{\emptyset \ne S \subseteq [3]} \sum_{i\in S} \alpha_i \beta_i \ne 0$.
Let $\beta^*=\max_i|\beta_i|$ and $M=(\sum_i|\alpha_i|)n^cd$.
We set $\perp_i = \beta_iM$ and $(n'')^c\geq \max\{2\beta^*M,(en)^c\}$.
Clearly the above choice satisfies all the lower bounds for $n''$, so now we need to show the last required property.
Suppose there exists set $\emptyset\ne S\subseteq[3]$ and elements $z_j\in[-dn^c,dn^c]$ for $j\in T=[3]\setminus S$ such that $\sum_{i\in S} \alpha_i\perp_i + \sum_{i\in T} \alpha_i z_i = t$.
By the properties of $\bar\beta$, $\sum_{i\in S}\alpha_i\perp_i=\sum_{i\in S}\alpha_i\beta_iM\ne 0$, so $|\sum_{i\in S}\alpha_i\perp_i|\geq M$.
As $|\sum_{i\in T}\alpha_iz_i|\leq \sum_{i\in T}|\alpha_i||z_i|\leq \sum_i|\alpha_i|n^cd\leq M/2$, we obtain $\sum_{i\in S}\alpha_i\perp_i + \sum_{i\in T}\alpha_i z_i \geq M/2 >t$, hence contradiction.
\end{proof}

\noindent
Intuitively, this lemma allows us to transform an instance of \Conv\ to a larger instance (at least by a factor of $e$) in which we can use the dummy elements $\perp$ that can never be part of a solution and the original non-$\perp$ elements  can be multiplied by $d$ and still fit within the universe $[-(n'')^c,(n'')^c]$.
With all the properties at hand, we are ready to state the equivalence between 1- and 3-partite variants of $\Convolution$.
We start with the adaptation of Lemma~\ref{le:3-partite-equivalent}.

\begin{lemma}\label{le:3-partite-conv-equivalent}
For every $c\geq1$ all non-trivial 3-partite variants of \Conv$_c$ are subquadratic-equivalent.
\end{lemma}
\begin{proof}
We can modify the coefficients of the instance as in Lemma~\ref{le:3-partite-equivalent} but by updating both the values and positions of elements simultaneously in the same way.
For instance, in the reduction $\Convolution\ldt 3\alpha 0\red \Convolution\ldt 3\alpha t$ we create the arrays $A'_i[k+y_i]=A_i[k]+y_i$, where $y_i$ are from the extended Euclidean algorithm and satisfy that $\sum_i\alpha_iy_i=t$.
Similarly, to show $\Convolution\ldt 3\alpha 0\red \Convolution\ldt 3\beta 0$ we define $q=\lcm(\beta_1,\beta_2,\beta_3)$ and set $A'_i[k \frac{\alpha_iq}{\beta_i} ] = A_i[k]\frac{\alpha_iq}{\beta_i}$.
As argued in Lemma~\ref{le:3-partite-equivalent}, in the reductions the elements and their indices increase at most by the constant factor $w$,
so in order to obtain an instance of \Conv$_c$ we first apply Lemma~\ref{le:perp} with $e=d=w$ and then apply the reduction.
\end{proof}

\begin{theorem}\label{thm:conv_1_equiv_3}
For all $c\geq1$ and non-trivial coefficients $\bar\alpha$ and $t$ we have: \Conv$_c(1,\bar\alpha,t) \eqq $ \Conv$_c(3,\bar\alpha,t)$.
\end{theorem}
\begin{proof}
 Recall that $n'=\frac{n-1}{2}$ and the array is indexed with $[-n',n']$. Assume that $\alpha$ and $t$ specify non-trivial instances of \Conv, i.e. $\alpha_i \neq 0$ for all $i \in [3]$ and $\gcd(\alpha_1, \alpha_2, \alpha_3) \mid t$.
 As all the reductions preserve the parameter $c$ for the size of the universe, we omit it while writing e.g. $\Convolution\ldt 1 \alpha t$.
 First, we show $\Convolution\ldt 1 \alpha t \red \Convolution\ldt 3 \alpha t$.
 As in Lemma~\ref{lm:from_1_to_3}, we  use color-coding from Fact~\ref{fact:color-coding}, with the domain of functions shifted by $n'+1$,  to determine to which array we should add the $k$-th element from the input array and then use $\perp$ from Lemma~\ref{le:perp} (with $d=e=1$) for all the remaining empty fields in the arrays.
 
 In the reduction $\Convolution\ldt 3 \alpha t\red \Convolution\ldt 1 \alpha t$ we first seek solutions containing a repeated index (solutions with $|\{j_1,j_2,j_3\}|<3$), in $\Oh(n)$ time.
 Now we need to find solutions with all the indices distinct.
 Again we use color-coding from Fact~\ref{fact:color-coding} obtaining $\Oh(\log^2 n)$ functions $\ind\rightarrow [3]$ and we consider each of them separately.
 For a single coloring function $\chi$ we set $A[k]=A_{\chi(k)}[k]$ for each $k\in\ind$.
 Then we use the approach of Lemmas~\ref{le:from_3_to_1_nonzero_t} and~\ref{le:from_3_to_1_zero_t} to modify the entries of the array in such a way that every triple of elements forming a solution consists of elements from three distinct arrays.
 We consider two cases:
 
\underline{\bf Case of $t\ne 0$.} By the extended Euclidean algorithm, we can choose an integer triple $\bar y$ such that $\sum_i \alpha_iy_i=t$ and apply Lemma~\ref{le:finding-gammas} on $\bar\alpha$ to obtain a triple $\bar\gamma$ of nonzero coefficients such that $\sum_i\alpha_i\gamma_i=0$ and for every non-constant forbidden origin~$f$ we have $\sum_i\alpha_i\gamma_{f(i)}\ne 0$. Let $C = 1+(\max_i\max\{|\gamma_i|,|y_i|\})\cdot \sum_i |\alpha_i|$. Given the arrays $A_1,A_2,A_3$, we construct the array $A$ by setting for each $k\in\ind$:
$$A[k]= C^2(A_{\chi(k)}[k]-y_{\chi(k)}) + C\gamma_{\chi(k)} + y_{\chi(k)}$$

Similarly to Lemma~\ref{le:from_3_to_1_nonzero_t}, it follows that every valid solution of $\Convolution\ldt1\alpha t$ in $A$ is a valid solution of $\Convolution\ldt3\alpha t$ on $A_1,A_2,A_3$ and vice versa.
As the elements can increase at most by factor $2C^2$, in order to obtain an instance of $\Convolution\ldt 1 \alpha t$ with elements within the universe, we first run Lemma~\ref{le:perp} with $e=1,d=2C^2$.

\underline{\bf Case of $t=0$.} Suppose first that $\sum_i \alpha_i \neq 0$. As $t=0$, in the above construction we do not need the part corresponding to the values $y_i$ from extended Euclidean algorithm, so we can construct the array $A$ by setting for each $k\in\ind$:
 \begin{equation}\label{eq:construct_A[k]_from_3_to_1}
  A[k]= CA_{\chi(k)}[k] + \gamma_{\chi(k)}.
 \end{equation}

By the properties of $\bar\gamma$, this is enough for the case when $\sum_i\alpha_i\ne 0$, as every valid solution of $\Convolution\ldt1\alpha t$ in $A$ is a valid solution of $\Convolution\ldt3\alpha t$ on $A_1,A_2,A_3$ and vice versa.
The case of $\sum_i \alpha_i=t=0$ is more involved and we consider it in two separate lemmas.

\begin{lemma}\label{le:Behrend-free-array}
 For all non-trivial coefficients $\bar\alpha$ such that $\sum_i\alpha_i=0$ and every $n>0$ we can calculate in $\Oh(n\cdot2^{\Oh(\sqrt{\log n})})$ time an array $\BB$ of $n$ elements from $[n]$, indexed by $P=[-\frac{n-1}2,\frac{n-1}2]$, such that for every three pairwise distinct indices $\bar j\in P^3$ such that $\sum_i \alpha_i j_i=0$ we have $\sum_i \alpha_i \BB[j_i]\ne 0$.
\end{lemma}
\begin{proof}
 Without loss of generality, assume $\alpha_1,\alpha_2>0$.
 Let $P=\ind$ be the set of all indices of the arrays.
 By Theorem~\ref{thm:partition_with_Behrends}, we partition $P$ into $e=2^{\Oh(\sqrt{\log n})}$ pairwise disjoint $(\alpha_1,\alpha_2)$-free sets $P_{1},P_{2},\ldots,P_{e}$ such that $P= \bigcup_j P_{j}$ in $\Oh(n\cdot2^{\Oh(\sqrt{\log n})})$ time.
 Let position $p\in P$ belong to set $P_{x_p}$, $\textsc{Bin}(x)$ consists of all positions containing 1 in binary representation of $x$ and $q=\sum_i|\alpha_i|+1$.
 Then we set $\BB[p]=\sum_{k\in \textsc{Bin}(x_p)} q^k$.

 Clearly the elements from $\BB$ are positive and upper bounded by $q^{\lfloor \log e\rfloor+1} = q^{\Oh(\sqrt{\log n})} = 2^{\Oh(\sqrt{\log n})} \leq n$.
 Now we need to show that for all triples $\bar j\in P^3$ of pairwise distinct elements such that $\sum_i\alpha_ij_i=0$ it holds that $\sum_i\alpha_i\BB[j_i]\ne 0$.
 By the choice of sufficiently big $q$:
 $\sum_i \alpha_i\BB[j_i]=0$ if and only if $\forall_{0\leq k \leq \lfloor\log e\rfloor} \sum_{i: k\in \textsc{Bin}(x_{j_i})} \alpha_i =0$.
 As $\sum_i \alpha_i=0$ and $\alpha_i\ne 0$ for $i\in [3]$, for every $S\subseteq[3]$ we have that $\sum_{i\in S} \alpha_i=0$ if and only if $S=\emptyset \vee S=[3]$ which translates to the fact that positions $x_{j_1},x_{j_2},x_{j_3}$ either all have $2^k$ in their binary representation or none of them has.
 Hence $\sum_i\alpha_i\BB[j_i]=0$ iff $x_{j_1},x_{j_2},x_{j_3}$ have exactly the same binary representation, so $x=x_{j_i}$ for all $i\in[3]$.
 This means that positions $j_1,j_2,j_3$ satisfying that $\sum_i \alpha_i j_i=0$ belong to the same $(\alpha_1,\alpha_2)$-free set $P_x$, which is a contradiction.
\end{proof}

\begin{lemma}
 For all  $c\geq1$ and non-trivial coefficients $\bar\alpha,t$ such that $\sum_i\alpha_i=t=0$ it holds $\Convolution\ldtc 3 \alpha 0 \red \Convolution\ldtc 1 \alpha 0$.
\end{lemma}
\begin{proof}

 Similarly to Lemma~\ref{le:from_3_to_1_zero_t}, we filter out elements from each array using Behrend's set, but now by indices of the values.
 More precisely, without loss of generality, assume $\alpha_1,\alpha_2>0$ and we consider $(\alpha_1,\alpha_2)$-free sets.
 Let $P=\ind$ be the set of all indices of the arrays.
 By Theorem~\ref{thm:partition_with_Behrends}, we partition $P$ into $e=2^{\Oh(\sqrt{\log n})}$ pairwise disjoint $(\alpha_1,\alpha_2)$-free sets $P_{1},P_{2},\ldots,P_{e}$ such that $P= \bigcup_j P_{j}$ in $\Oh(n\cdot2^{\Oh(\sqrt{\log n})})$ time.
 For each $\bar u\in [e]^3$, we restrict each array $A_i$ only to elements from positions $P_{u_i}$ and construct an instance of $\Convolution\ldt1\alpha t$, similarly as in Equation~\eqref{eq:construct_A[k]_from_3_to_1}.
 
 However, now we cannot directly use the dummy symbols $\perp_i$ from Lemma~\ref{le:perp} as we might consider three positions $j_1,j_2,j_3$ of the array $A$ such that their elements originate from the same array~$A_\mu$, where $\mu=\chi(j_1)=\chi(j_2)=\chi(j_3)$ and they were all discarded from the array $A_\mu$, that is $j_i\notin~P_{u_\mu}$ for $i\in[3]$.
 Then using $\perp_\mu$ for each of them will result in a triple of elements satisfying $\sum_i \alpha_i A[j_i]=\sum_i \alpha_i \perp_\mu = 0$ that does not correspond to a valid solution of $\Convolution\ldt3\alpha t$.
 To overcome this issue, we will replace the dummy elements with elements of array $\BB$ from Lemma~\ref{le:Behrend-free-array}.
 
 Let $M=\sum_i|\alpha_i|n^c$, $n''$ satisfy $(n'')^c \geq 3CM$ and $\BB$ be an array defined in Lemma~\ref{le:Behrend-free-array} for coefficients $\bar\alpha$, consisting of $n''$ elements indexed by $[-\frac{n''-1}2,\frac{n''-1}{2}]$.
 As $c\geq1$, all elements from~$\BB$ are smaller than $M$.
  While applying Equation~\eqref{eq:construct_A[k]_from_3_to_1}, we replace the filtered out elements with a transformation of the corresponding elements from $\BB[k]$:
 \begin{equation}\label{eq:construct_A[k]_from_3_to_1_with_Behrend}
  A[k]= \begin{cases}
CA_{\chi(k)}[k] + \gamma_{\chi(k)} & k\in P_{u_{\chi(k)}} \wedge k\in [-\frac{n-1}2,\frac{n-1}2] \\
C(2M+\BB[k]) + \gamma_{\chi(k)} & \text{otherwise}
        \end{cases}
 \end{equation}
 
We say that an element $A[k]$ is discarded if $k\notin P_{u_{\chi(k)}}$ or $k\notin [-\frac{n-1}2,\frac{n-1}2]$.
It is easy to see that any valid solution of $\Convolution\ldt3\alpha t$ on $A_1,A_2,A_3$ gives a valid solution of $\Convolution\ldt1\alpha t$ in $A$ for some choice of $\bar u\in [e]^3$.
Now we show the opposite direction.
Consider a triple $\bar j$ of positions that form a solution of $\Convolution\ldt1\alpha 0$ for a particular choice of $\bar u\in [e]^3$.
This means that $j_i$ are pairwise distinct, $\sum_i\alpha_ij_i=0$, $\sum_i\alpha_iA[j_i]=0$.
Let $S=\{i:j_i\in P_{u_i} \wedge i \in [3] \}$ be the set of indices of elements originating from a non-discarded element, and $T=[3]\setminus S$.
Then:
\begin{align*}
0 = \sum_i\alpha_iA[j_i] &= \sum_{i\in S}\alpha_i(CA_{\chi(j_i)}[j_i] + \gamma_{\chi(j_i)}) + \sum_{i\in T}\alpha_i(C(2M+\BB[j_i])+\gamma_{\chi(j_i)}) \\
&= \sum_i\alpha_i\gamma_{\chi(j_i)} + C\left(\sum_{i\in S}\alpha_iA_{\chi(j_i)}[j_i] + \sum_{i\in T}\alpha_i\left(2M+\BB[j_i]\right)  \right)
\end{align*}
As the above expression equals 0 we have $\sum_i\alpha_i\gamma_{\chi(j_i)}=0$, so by Lemma~\ref{le:finding-gammas}, $(\chi(j_1),\chi(j_2),\chi(j_3))$ is either a constant ($\mu=\chi(j_1)=\chi(j_2)=\chi(j_3)$) or not forbidden origin.
Next, the coefficient by $C$ also equals $0$, so let $(\star) = \sum_{i\in S}\alpha_iA_{\chi(j_i)}[j_i] + \sum_{i\in T}\alpha_i\left(2M+\BB[j_i]\right) = 0$.
There are three cases to consider depending on the number of indices corresponding to non-discarded elements:
\begin{itemize}
\item[$|S|=3$:] We have $0=(\star)=\sum_i\alpha_iA_{\chi(j_i)}[j_i]=0$.
If all $\chi(j_i)$ are equal to some $\mu$, all elements originate from array $A_\mu$.
However, we restricted $A_\mu$ only to elements on positions from $(\alpha_1,\alpha_2)$-free set $P_\mu$, which contradicts the fact that $\sum_i \alpha_ij_i=0$.
Otherwise, $(\chi(j_1),\chi(j_2),\chi(j_3))$ is a non-forbidden origin, so we obtain a valid solution to the original instance of $\Convolution\ldt3\alpha t$.

\item[$|S|=0$:] We have $0=(\star)= \sum_i\alpha_i(2M+\BB[j_i]) = \sum_i \alpha_i\BB[j_i]$.
which contradicts the fact that this sum is non-zero, by Lemma~\ref{le:Behrend-free-array}.
\item[else:] For $|S|\in\{1,2\}$ we have $|T|\in\{1,2\}$.
Let $(\star_1) = \sum_{i\in T}2M\alpha_i= 2M\sum_{i\in T}\alpha_i $ so $|(\star_1)|\geq2M$ because $\alpha_i\ne 0$ (as we consider non-trivial coefficients $\bar\alpha$) and $\sum \alpha_i=0$ (by the assumption).
Next, let $(\star_2)=(\star)-(\star_1)= \sum_{i\in S} \alpha_iA_{\chi(j_i)}[j_i] + \sum_{i\in T}\alpha_i\BB[j_i]$ and then $|(\star_2)| \leq \sum_i |\alpha_i|n^c \leq M$ because elements from $A$ and $\BB$ are from $[-n^c,n^c]$.
Thus $|(\star)| = |(\star_1) + (\star_2)| \geq M$ which contradicts that $(\star) =0$.

\end{itemize}
\noindent
Finally, all the elements from $A$ are from $[-(n'')^c,(n'')^c]$, because $A_{\chi(k)}[k]\in[-n^c,n^c]$, $2M+\BB[k] < 3M$ and $(n'')^c\geq 3CM = 3C\sum_i|\alpha_i|n^c$, so the obtained array is an instance of $\Convolution\ldt1\alpha 0$.
To summarize, we created $e^3$ such instances of $\Convolution\ldt1\alpha 0$ with arrays of $n''=\Oh(n)$ elements, where $e=2^{\Oh(\sqrt{\log n})}$.
If one could solve $\Convolution\ldt1\alpha 0$ in $\Oh((n'')^{2-\eps})$ time, combining that with our reduction will give an algorithm solving an instance of $\Convolution\ldt3\alpha 0$ in time $\Oh(e^3(n'')^{2-\eps}) = \Oh(n^{2-\delta})$ for every $0<\delta<\eps$.
\end{proof}

This concludes the subquadratic reduction from an instance of $\Conv_c(3,\bar\alpha,0)$ to $2^{\Oh(\sqrt{\log n})}$ instances of $\Conv_c(1,\bar\alpha,0)$.
\qedhere
\end{proof}

\section{Equivalences between 3LDT and Conv3LDT}\label{se:equiv-3ldt-conv4ldt}

In this section, we show reductions between 3-partite variants of 3\LDT{} and \Conv{} (see Figure~\ref{fig:convreductions} for an overview).
We start with a folklore reduction $\Convolution\ldtc 3 \alpha 0 \red 3\LDT_{c+1}(3,\bar\alpha,0)$ that increases the size of the universe by $n$ \cite{ChanH20}.

\begin{figure}[h]
    \centering
    \begin{tikzpicture}
    \begin{scope}[every node/.style={align=center},
                  every edge/.style={draw=black}]
        
        \node (n12) at (-0.2,0) {$\Convolution\ldt 1 {\alpha_1} {t_1}$};
        
	\node (n21) at (4.5,-1.5) {\Convolution3\SUM};
        \node (n22) at (4.5,0) {$\Convolution\ldt 3 {\alpha_1} {t_1}$};
        \node (n23) at (4.5,1.5) {$\Convolution\ldt 3 {\alpha_2} {0}$};

        \node (n31) at (9,-1.5) {3\SUM};
        \node (n32) at (9,0) {$\ldt 3 {\alpha_3} {t_3}$};
        \node (n33) at (9,1.5) {$\ldt 3 {\alpha_2} {0}$};
        
        \node (n42) at (13,0) {$\ldt 1 {\alpha_4} {t_4}$};
        
        \path [<->] (n12) edge node[below] {\tiny{Thm. \ref{thm:conv_1_equiv_3}}} (n22);
        \path [->] (n22) edge node [left] {\tiny{Lem.}~\ref{le:3-partite-conv-equivalent}} (n23);
        \path [->] (n21) edge node [left] {\tiny{Lem.}~\ref{le:3-partite-conv-equivalent}} (n22);
        
        \path [<-] (n32) edge node [left] {\tiny{Lem.}~\ref{le:3-partite-equivalent}} (n33);
        \path [<-] (n31) edge node [left] {\tiny{Lem.}~\ref{le:3-partite-equivalent}} (n32);
        
        \path [->] (n23) edge  node[below] {\tiny{Lem. \ref{le:conv_to_nonconv}}} node[above]{\scriptsize{$c\rightarrow c+1$}} (n33);
        \path [<-] (n21) edge node [below] {\tiny{Thm. \ref{thm:convolution}}} node[above]{\scriptsize{$c-1\leftarrow c$}} (n31);
        
        \path [<->] (n32) edge node[below] {\tiny{Thm.~\ref{thm:everything-equivalent}}} (n42);
        
    \end{scope}
    \end{tikzpicture}
    \caption{Subquadratic reductions between different variants of \Conv{} and 3\LDT.
    Subscripts to coefficients $\alpha$ and $t$ denote if the reduction allows changing the coefficients or not.
    Only two reductions change the value of parameter $c$, all other reductions preserve it.
}
    \label{fig:convreductions}
\end{figure}

\begin{lemma}\label{le:conv_to_nonconv}
 For all $c\geq1$ and $\bar\alpha$, an instance of $\Convolution\ldtc 3 \alpha 0$ reduces in linear time to an instance of $3\LDT_{c+1}(3,\bar\alpha,0)$.
\end{lemma}
\begin{proof}
 We create sets $S_i=\{A_i[j]\cdot W + j : j\in \ind\}$ where $W=\left(\sum_i|\alpha_i| + 1 \right )\cdot n$.
 By the choice of~$W$, a solution $\bar x$ where $x_i\in S_i$ satisfies $\sum_i\alpha_ix_i=0$ iff both $\sum_i\alpha_i (x_i \bmod W)= 0$ and $\sum_i\alpha_i \lfloor\frac{x_i}{W}\rfloor = 0$, which exactly corresponds to the condition on both the values and indices of elements from the instance of $\Convolution\ldt3\alpha 0$.
 The elements from $S_i$ are bounded by $W\cdot n^c + \frac{n}{2} \le n^c \cdot n \cdot (\sum_i|\alpha_i| + 2)$ so we let $n''=n\lceil(\sum_i|\alpha_i|+3)^{1/(c+1)}\rceil$.
 Then we have an instance of $3\LDT_{c+1}${} with sets of at most $n''$ elements from $[-(n'')^{c+1},(n'')^{c+1}]$.
 By Lemma~\ref{le:handle_semi_instance}  we can reduce it to an instance of $3\LDT_{c+1}(3,\bar\alpha,0)$ on sets with $n'''=\Oh(n'')=\Oh(n)$ elements.
\end{proof}
Next, we show how to reduce 3\LDT{} to \Conv{} while controlling the exponent in the size of the universe.
The starting point is the celebrated proof by \Patrascu{} who showed a randomized reduction from 3\SUM{} to \Convolution 3\SUM{} (without controlling the size of the universe).
We follow a later deterministic approach of Chan and He \cite{ChanH20}, except that we need to tweak it to control the size of the universe and extend it to arbitrary variants of 3\LDT.

Recall that in the 3-partite variant of 3\SUM{} we are given three sets $S_1,S_2,S_3$ and need to check if there are three numbers $x_1,x_2,x_3$ such that $x_1+x_2=x_3$ and $x_i\in S_i$ for $i\in[3]$.
Earlier we defined 3\SUM{} for universe $[-n^3,n^3]$, but in this section we consider arbitrary universe $[n^c,n^c]$ for $c\geq2$, that is variants $3\LDT_c(3,(1,1,-1),0)$ for any $c\geq2$.
Similarly, $\Convolution3\SUM{}$ denotes a family of variants $\Conv_c(3,(1,1,-1),0)$ for any $c\geq1$.

\begin{theorem}\label{thm:convolution}
 For all $c\geq2$, an instance of 3\SUM{}$_c$ on 3 sets with $n$ numbers reduces in $n\polylog n$ time to $\polylog n$ instances of \Convolution 3\SUM{}$_{c-1}$ on arrays of $\Oh(n)$ elements.
\end{theorem}
\begin{proof}
 
 We define a semi-instance of 3\SUM{} on three sets $S_i$ of at most $n$ elements from $[-n^c,n^c]$ to be $m$-\textit{spread} if for all $i\in[3]$ we have $|\{x \bmod m: x\in S_i\}|=|S_i|$, that is all elements of each set have distinct remainders modulo $m$.
 Note that this implies that the sets have at most $m$ elements each. 
 An important tool in our reduction is the deterministic recursive algorithm for reducing 3\SUM{} to a number of instances of \Convolution 3\SUM{} described by Chan and He \cite{ChanH20}.
 We slightly reformulate their algorithm:
 
 \begin{lemma}[Section 5 \cite{ChanH20}]\label{le:3sum_to_spread}
 There exists a deterministic reduction $\mathcal{R}$ such that for every instance~$\II$ of 3\SUM{}$_c$ on three sets of $n$ elements and parameter $t$ satisfying that $t\geq \polylog n$, $\mathcal{R}$ reduces $\II$ to $k=t^3\polylog n$ instances $\II_1,\ldots,\II_k$ of 3\SUM{} on numbers from $[-n^c,n^c]$ where the $j$-th instance $\II_j$ is accompanied with a number $m_j\in[\frac n4,n)$ such that $\II_j$ is $m_j$-spread.
  The reduction runs in $\Oh(n^{1.5}\polylog n)$ time.
 \end{lemma}
 \begin{proof}
  We provide an informal overview of the algorithm of Chan and He.
  Given an instance of 3\SUM{} with three sets $S_i$, we find a number $m$ and hash elements of every set into buckets modulo~$m$.
  Let $t$ be a parameter.
  We call an element \textit{bad} if its bucket contains more than $t$ %\todob{they had t+1}
  elements, otherwise \textit{good}.
  The number $m\in\Theta(n)$ will be chosen in $n^{1.5}\polylog n$ time in such a way that for every $i\in[3]$ the total number of bad elements in $S_i$ is at most $d_1\frac{|S_i|^2}{mt}\log^{d_2} n$ for some constants $d_1,d_2>0$.\footnote{The choice of $m$ with desired properties is described in detail in Section 5 of \cite{ChanH20}.
  Our only modification is to set $m_1=m_2=\sqrt n$ which gives an algorithm for finding $m$ in $n^{1.5}\polylog n$ time.} In order to find a triple $\bar x$ such that $x_i\in S_i$ for $i\in[3]$ and $x_1+x_2=x_3$ we proceed as follows:
  \begin{itemize}
   \item First we find solutions in which all elements $x_i$ are good. We consider $t^3$ triples $\bar p\in[t]^3$ and for each of them we restrict the problem as follows. For each $i\in[3]$ we only keep the $p_i$-th element from each of the buckets of set $S_i$.
   This gives us $t^3$ $m$-spread instances of 3\SUM{}.
   \item In order to find a solution with at least one bad element $x_i$, e.g. $x_3$, we call the procedure recursively to solve 3\SUM{} for the following sets: $S_1,S_2$ and the third set consisting of only the bad elements from $S_3$. We proceed similarly for the case when $x_1$ or $x_2$ is bad.
  \end{itemize}
  We require that $t\geq 8 d_1\log^{d_2}n=\polylog n$ and choosing $m\in[\frac{n}{4},n)$ will give us the number of bad elements in set $S_i$ at most $d_1\frac{|S_i|^2}{mt}\log^{d_2} n \leq\frac{|S_i|^2}{2n}$.
  Suppose that a set was repeatedly replaced with the set of its bad elements and after the $j$-th recursive step its size is $\frac{n}{k_j}$. As $(\frac{n}{k_j})^2\frac{1}{2n}=\frac{n}{2{k_{j}}^2}$, we have $k_{j+1}\geq {2k_j^2}$. By setting $k_0=1$, we therefore obtain $k_j\geq 2^{2^{j}-1}$.
  Thus the depth of the recursion tree is $h=\Oh(\log \log n$) and each node makes at most 3 recursive calls by restricting only to bad elements from $S_i$ for $i\in[3]$.
  Hence in total there are at most $t^33^h=t^3\polylog n$ instances of $m$-spread 3\SUM{}, where the values of $m$ may differ between different nodes of the recursion tree.
 \end{proof}
\noindent
 Now we show how to reduce a $m$-spread instance of 3\SUM{} to an instance of \Convolution 3\SUM{}, decreasing the size of the universe by a factor of $m$.
 
 \begin{lemma}
  Let $m\in[\frac n4,n)$. We can reduce in linear time a $m$-spread instance of 3\SUM{} with numbers from $[-n^c,n^c]$ to an instance of \Convolution3\SUM{} with arrays on $n''=\Oh(n)$ elements from $[-(n'')^{c-1},(n'')^{c-1}]$.
 \end{lemma}
 \begin{proof}
 Let $\II^m$ be the $m$-spread instance of 3\SUM{}, $h(x) = x \bmod m$ and $v(x) = \left\lfloor \frac xm \right\rfloor$.
 We create the instance $\II$ of \Convolution 3\SUM{} with arrays $A_i$ indexed with $[-2m+1,2m-1]$ where the $i$-th of them is initially filled with $\perp_i$ defined in Lemma~\ref{le:perp}.
 For $i\in[3]$, for all $x\in S_i$ we set $A_i[h(x)]=v(x)$.
 Additionally, for all $z\in S_3$ we set $A_3[h(z)+m]=v(z)-1$.
 Note that no two elements were put in the same cell as the sets $S_i$ are $m$-spread.
 
 Now we show that this reduction is correct.
 Suppose there is a triple $\bar x:x_i\in S_i$ for $i\in[3]$ and $x_1+x_2=x_3$.
 There are two cases two consider:
 \begin{itemize}
  \item if $h(x_1)+h(x_2)<m$, we have $v(x_1)+v(x_2)=v(x_3)$ and $h(x_1)+h(x_2)=h(x_3)$, so the triple $(h(x_1),h(x_2),h(x_3))$ is a solution for $\II$,
  \item if $h(x_1)+h(x_2)\geq m$, we have $v(x_1)+v(x_2)=v(x_3)-1$ and $h(x_1)+h(x_2)=h(x_3)+m$, so the triple $(h(x_1),h(x_2),h(x_3)+m)$ is a solution for $\II$.
 \end{itemize}
 In the other direction, suppose there is a solution $\bar y$ for $\II$ and we show that it gives a solution for $\II^m$.
 Indeed, by definition we have $y_1+y_2=y_3$ and $A_1[y_1]+A_2[y_2]=A_3[y_3]$ so $y_1+y_2+m(A_1[y_1]+A_2[y_2])=y_3+mA_3[y_3]$.
 From the properties of $\perp_i$, every solution of $\II$ contains only the non-$\perp$ values, so every $y_i$ must be obtained from a value $x_i\in S_i$.
 We set $x_i=y_i+mA_i[y_i]$ for $i\in[3]$ and claim that it is a solution for $\II^m$.
 Clearly, for $i=1,2$ it holds that $x_i\in S_i$.
 For $i=3$, $x_3=y_3+mA_3[y_3]=(y_3-m)+m(A_3[y_3]+1)$, so $x_3\in S_3$ both for $y_3<m$ and $y_3\geq m$.
 
 Recall that the $m$-spread instance $\II^m$ consists of sets of at most $m\leq n$ elements from $[-n^c,n^c]$.
 Then $v(x)\in[-\frac{n^c}m,\frac{n^c}m]$ for $x\in S_i$ and $i\in[3]$ and the created arrays have elements on the positions from $[-(2m-1),2m-1]$.
 As we need $\perp_i$ for the non-occupied elements of the arrays, we apply Lemma~\ref{le:perp} for $n$, $d=4$ (as we need $\frac{n^c}{m} \leq dn^{c-1}$) and $e=4$ (as the obtained array has $4m-1\leq 4n$ elements).
 This gives us $n''= \Oh(n)$ and $\bar\perp$ such tat we can have an array of $n''$ entries, unoccupied elements from the $i$-th array are replaced with $\perp_i$ and all entries (occupied or not) are from $[-(n'')^{c-1},(n'')^{c-1}]$.
 Hence the claim follows.
 \end{proof}
 
 We set $t=\polylog n$ such that it satisfies the requirement from Lemma~\ref{le:3sum_to_spread}.
 Combining the two above lemmas we obtain a reduction from 3\SUM{} on 3 sets with~$n$ numbers from $[-n^c,n^c]$ to $k=t^3\polylog n=\polylog n$ instances of \Convolution 3\SUM{} on $3$ arrays on $n''=\Oh(n)$ elements from $[-n''^{(c-1)},n''^{(c-1)}]$, which concludes the proof.
\end{proof}
%Note that a slightly different variant of this reduction is possible, following the idea of Kopelowitz et al. \cite{KopelowitzPP16} that creates $T\cdot \Oh(L)$ instances instead of $T^3\cdot \Oh(L)$, by putting whole buckets at once and only choosing the cyclic shift of elements from the third bucket. Similarly, we could obtain a direct reduction from 3\SUM{} to the 1-partite instance of \Convolution 3\SUM{} by considering an array on $8n$ entries and inserting elements from the first array into positions congruent to 1 modulo 8, from the second to positions congruent to 3 and from the third array to positions congruent to 4 modulo~8.
\noindent
By Theorem~\ref{thm:convolution} and Lemma~\ref{le:conv_to_nonconv}, we immediately obtain the following corollary:
\begin{corollary}\label{cor:3sum_equiv_conv_3sum}
For all $c\geq2$ it holds $3\LDT_c(3,(1,1,-1),0) \eqq \Conv_{c-1}(3,(1,1,-1),0)$.
\end{corollary}
  \noindent
Now we can combine Lemma~\ref{le:conv_to_nonconv}, Theorem~\ref{thm:everything-equivalent}, Theorem~\ref{thm:convolution}, Lemma~\ref{le:3-partite-conv-equivalent} and Theorem~\ref{thm:conv_1_equiv_3} as presented in Figure~\ref{fig:convreductions} to obtain the following theorem:

\AllLDTEquivalences*

Finally, we discuss the subtlety of the negative indices in the arrays that we introduced in the definition of \Conv.
A natural question is whether we can achieve a similar result with only positive indices in \Conv. 
Consider a non-trivial variant of \Conv{} with all $\alpha_i>0$. There is a constant number of triples of indices $j_1,j_2,j_3\in [n]$ such that $\sum_i \alpha_i j_i = t$ which we can all check in linear-time.
Similarly for the variants with all $\alpha_i<0$.
Therefore, in those cases there is a fast algorithm for 3\SUM{} which makes such an equivalence unlikely.

For that reason, in the definition of $\Conv$ we allow negative indices of arrays.
However, if not all coefficients $\alpha$ have the same sign, we can work with instances of arrays with only positive indices: 

\begin{lemma}
 For all $c\geq1$, coefficients $\bar\alpha$ not all having the same sign, and arbitrary $t$, every non-trivial variant of $\Conv_c(3,\bar\alpha,t)$ is equivalent to the same variant of $\Conv_c(3,\bar\alpha,t)$, but on arrays with only positive indices.
\end{lemma}
\begin{proof}
 Suppose the original instance operates on arrays $A_i\ind$.
 Let $\delta=n'+1$ and we show how to choose~$\bar\beta$ that satisfy $\sum_i \alpha_i\beta_i=0$ and $\beta_i>0$ for all $i\in[3]$.
 As not all coefficients $\bar\alpha$ have the same sign, without loss of generality we can assume that $\alpha_1>0$ and $\alpha_2,\alpha_3<0$. 
 Then it suffices to set $\beta_1=-\alpha_2-\alpha_3$ and $\beta_2=\beta_3=\alpha_1$ that satisfy the above properties. 
 
 We shift the $i$-th array by $\delta \beta_i$ elements, that is we set $A_i'[k]=A_i[k-\delta \beta_i]$.
 Then any triple~$\bar j'$ of indices in $A'$ corresponds to indices $\bar j - \delta\bar\beta$ in $A$ and we have
 $\sum_i\alpha_ij_i'=\sum_i\alpha_i(j_i-\delta\beta_i)=\sum_i\alpha_ij_i - \delta(\sum_i\alpha_i\beta_i) = \sum_i\alpha_ij_i$.
 All the elements from arrays $A_i$ are moved to entries of $A_i'$ with indices from $[1,n'']$ where $n''=n'+\delta\max_i\beta_i=\Oh(n')$ so there is a solution of $\Convolution\ldt3\alpha t$ for arrays $A$ iff there is a solution for tables $A'$.
 Depending on the value of $\beta_i$, we have to fill either the first $\delta\beta_i-n'-1$ elements of $A_i'$ with $\perp_i$, the last $n''-(n'+\delta\beta_i)$ elements or both at the beginning and end of $A_i'$.
 To this end we apply Lemma~\ref{le:perp} with $d=1$ and $e=\lceil n''/n \rceil$ to obtain $n'''$ and $\bar\perp$ such that $\perp_1,\perp_2,\perp_3$ have the desired properties. Then we use $\bar\perp$ to fill the non-existent elements in the arrays and finally obtain arrays on $n'''$ elements from $[-(n''')^c,(n''')^c]$ on positions $[1,n''']$.
\end{proof}

\begin{corollary}
 For all $c\geq2$, all non-trivial 1- and 3-partite variants of 3\LDT{}$_c$ and \Conv{}$_{c-1}$ with non-trivial coefficients such that $\bar\alpha$ not all having the same sign are subquadratic-equivalent, even if we operate on arrays with only positive indices.
\end{corollary}

In particular, in \AVG{} we have $\bar\alpha=(1,1,-2)$, so not all the coefficients have the same sign and we can obtain a similar result for \AVG, which was considered by Erickson:

\begin{corollary}
 For all $c\geq2$, \AVG{}$_c$ is subquadratic-equivalent to \Convolution\AVG$_{c-1}$, even if we operate on  arrays with only positive indices.
\end{corollary}

\section{Reducing the size of the universe}\label{se:universe-reduction}

Fischer et al. \cite{FischerKP24} showed a deterministic reduction of the size of the universe for 3\SUM{} from arbitrarily big to cubic.
In this section we show how to extend their result to all 3-partite variants of 3\LDT.
Next, we show how to use reductions from the previous sections to show universe reduction for 1-partite 3\LDT{} and variants of \Conv.

At the input we are given $n$ numbers of length $\log U$.
For $U>2^{n^{\Omega(1)}}$, the trivial $\Oh(n^2\log U)$ algorithm is already subquadratic in the size of the input, so we can focus only on the cases when $U<2^{n^{o(1)}}$.
We will operate on numbers of length $\Oh(\log U)$: add them, subtract, multiply by a number at most $\Oh(n)$ or calculate modulo a small number.
All these operations take $\Oh(\log U)=n^{o(1)}$ time which is irrelevant from the perspective of subquadratic reductions.
For this reason we will not explicitly mention this factor in the time complexity of the presented reductions.

\begin{theorem}\label{thm:universe-reduction}
For all non-trivial coefficients $\bar\alpha$ and $t$, if an instance of $3\LDT_3(3,\bar\alpha,t)$ of sets of $n$ numbers over $[-n^3,n^3]$ can be solved in deterministic time $\Oh(n^{2-\eps})$ for some $\eps>0$, then $3\LDT_?(3,\bar\alpha,t)$ on sets of $n$ numbers from $[-U,U]$ can be solved in deterministic time $\Oh(n^{2-\eps'}\log^{c}U)$ for some constants $\eps',c>0$.
\end{theorem}
\begin{proof}
We show how to extend the proof of Theorem 2 from Fischer et al. \cite{FischerKP24} from 3\SUM{} to 3\LDT.
For that purpose we restate the relevant lemmas adjusted to the general, 3-partite case.
We do not restate the full proofs, we only highlight what needs to be changed.

Following \cite{FischerKP24}, we use the notation $\Ohtilda(T) = T(\log T)^{\Oh(1)}$.
To avoid clutter, in this section we implicitly allow instances of 3\LDT{} to have at most $n$ numbers, instead of exactly $n$ numbers.
Then, by applying Lemma~\ref{le:handle_semi_instance} we can transform each such instance to contain exactly $n$ elements from $[-n^3,n^3]$, as in the original definition.

\begin{lemma}[Lemma 11 \cite{FischerKP24}]\label{le:det_3sum_l11}
Let $\bar\alpha$ and $t$ be non-trivial coefficients and $0\leq\mu < 3, \delta>0$.
There is a deterministic algorithm that, given a $3\LDT_?(3,\bar\alpha,t)$ instance $A$ on sets $A_1,A_2,A_3$ of $n$ numbers over $[-U,U]$, either finds a positive modulus $m\in[n^\mu,2n^\mu)$ such that
$$ |\{(x_1,x_2,x_3)\in A_1\times A_2\times A_3: \sum_i \alpha_i x_i \equiv t \pmod m\}| \leq n^{3-\mu+\delta}(\log U)^{\Oh(1/\delta)}$$
or decides that $A$ is a yes-instance.
The algorithm runs in $\Ohtilda(n^{\max(\mu,1)+\delta})$ time.
\end{lemma}
\begin{proof}
In order to handle the more general case, there are only slight differences with respect to the original proof.
First, in our applications the number of instances $g$ always equals 1 so we can skip it.
Second, while computing $S(m\cdot p)$ with FFT we also need to take into account the coefficients~$\bar\alpha$.

Finally, we need to alter the analysis of the probability that a random prime $p$ divides $q=\sum_i\alpha_ix_i-t$ assuming that $q\ne 0$.
Let $r= 1+ \sum_i|\alpha_i|$.
As $t$ is constant, we have $q\in [-rU,rU]$.
Following the analysis from the original proof, the aforementioned probability is at most
$$ \frac{\log_{n^\delta}(rU)}{\Omega(n^\delta/\log n)} = \Oh(n^{-\delta}\log U)$$
\noindent
and the rest of the analysis of the correctness follows.
\end{proof}

\begin{lemma}[Lemma 12 \cite{FischerKP24}]\label{le:det_3sum_l12}
 Let $A_1,A_2,A_3\subseteq [-U,U]$ be sets of size $n$, and let $g\geq 1$.
 We can deterministically construct a partition of sets $A_1,A_2,A_3$ into subsets $A_i^1,\ldots,A_i^g$ for $i\in[3]$ where $A_i=\bigcup_{j\in[g]} A_i^j$, and a set of triples $R\subseteq[g]^3$ of size $\Oh(g^2)$ such that:
 \begin{itemize}
  \item Each set $A_i^{j_i}$ has size $\Oh(n/g)$ can be covered by an interval of length $\Oh(U/g)$.
  \item For all triples $x_1,x_2,x_3\in A_1\times A_2\times A_3$ with $\sum_i \alpha_ix_i = t$, there is some $(j_1,j_2,j_3)\in R$ with $x_i\in A_i^{j_i}$ for $i\in[3]$.
 \end{itemize}
 The algorithm runs in $\Ohtilda(n+g^2)$.
\end{lemma}
\begin{proof}
Our analysis is a slight generalization of the proof from \cite{FischerKP24}.
We only carefully describe how to alter the definitions in the original proof to handle all possible signs of coefficients in $\bar\alpha$.
First, similarly as in \cite{FischerKP24}, for ever $i\in [3]$ we partition set $A_i$ into subsets $A_i^1,\ldots,A_i^g$ such that for each $j\in [g]$:
\begin{itemize}
 \item $|A_i^j|\leq 2n/g$, and
 \item $\max(A_i^j)-\min(A_i^j)\leq 2U/g$.
\end{itemize}
We can find such a partition greedily, by iterating over the sorted elements of $A_i$ until either of the two required conditions is violated.
As every condition separately generates at most $g/2$ sets, in total we obtain at most $g$ sets.
We assume that the sets are in the natural order, that is $\max(A_i^j) < \min(A_i^{j+1})$ for all $j\in[g-1]$.

Now we show how to construct the set $R\subseteq [g]^3$ of $\Oh(g^2)$ triples such that for all triples $x_1,x_2,x_3\in A_1\times A_2\times A_3$ with $\sum_i \alpha_ix_i = t$, there is some $(j_1,j_2,j_3)\in R$ with $x_i\in A_i^{j_i}$ for all $i\in[3]$.
First, we add to $R$ all triples containing at most two distinct elements, there are $g+g\cdot(g-1)\cdot3 = \Oh(g^2)$ of them.
Now we need to add triples of three distinct elements.

Let $xA$ be the set $\{x\cdot a : a\in A\}$.
We construct the set $R$ as follows.
Suppose that $\alpha_3$ is positive.
For every pair $(i,j)\in [g]^2$ , we binary search the
smallest $k_1\in[g]$ such that $\max(\alpha_1 A_1^i) +\max(\alpha_2 A_2^j) + \max(\alpha_3 A_3^{k_3}) \geq t$ and the
largest $k_2\in [k_1,g]$ such that $\min(\alpha_1 A_1^i) +\min(\alpha_2 A_2^j) + \min(\alpha_3 A_3^{k_3}) \leq t$.
Then for all $k_1\leq k \leq k_2$ we add to $R$ the triple $(i,j,k)$.
The case of negative $\alpha_3$ is symmetric, we swap roles of $k_1$ and $k_2$.

Observe that the above approach is correct (we do not miss any relevant triple in $R$), it runs in $\Ohtilda(g^2+|R|)$ time and the size of $R$ is $\Oh(g^3)$.
By the above construction, we do not add to $R$ triples $(i,j,k)$ such that either:
\begin{enumerate}
 \item $\min(\alpha_1 A_1^i) +\min(\alpha_2 A_2^j) + \min(\alpha_3 A_3^k) > t$, or
 \item $\max(\alpha_1 A_1^i) +\max(\alpha_2 A_2^j) + \max(\alpha_3 A_3^k) < t$.
\end{enumerate}
Clearly, for such triples there is no triple $(x_1,x_2,x_3)\in A_1^i\times A_2^j \times A_3^k$ such that $\sum_i \alpha_i x_i =t$.
We call such a triple $(i,j,k)$ \emph{immediate}.
Now we show that the size of $R$ is actually $\Oh(g^2)$ as desired.

Similarly as in \cite{FischerKP24}, consider the partially ordered set $P=([g]^3,\prec)$ where $(i,j,k)\prec(i',j',k')$ if and only if $\max(\alpha_1A_1^i)<\min(\alpha_1A_1^{i'})$, $\max(\alpha_2A_2^j)<\min(\alpha_2A_2^{j'})$ and $\max(\alpha_3A_3^k)<\min(\alpha_3A_3^{k'})$.
Note that this is not the standard lexicographical order, as some $\alpha_i$ might be negative.
Observe that for two triples $(i,j,k), (i',j',k')$ such that $(i,j,k)\prec(i',j',k')$, at least one of them is immediate.
Indeed, suppose the contrary, that they both falsify conditions (1) and (2). Then:
\begin{align*}
 t &\leq \max(\alpha_1 A_1^i) +\max(\alpha_2 A_2^j) + \max(\alpha_3 A_3^{k}) &&\text{as } (i,j,k) \text{ falsifies (2)} \\
 & < \min(\alpha_1A_1^{i'})+\min(\alpha_2A_2^{j'})+ \min(\alpha_3A_3^{k'}) &&\text{as } (i,j,k)\prec(i',j',k') \\
 & \leq t &&\text{as } (i',j',k') \text{ falsifies (1)}
\end{align*}
\noindent
which is a contradiction.
Thus, on any chain (totally ordered set of elements) in $P$ there is at most one non-immediate triple in $R$.
Hence it suffices to show that we can cover $P$ with $\Oh(g^2)$ chains.
Let $\sgn(x)=\frac{x}{|x|}$\footnote{In this paragraph we only consider the sign of coefficients $\alpha$. They are always non-zero as the variant of 3\LDT{} that we consider is non-trivial.}.
For all $a,b\in [-(g-1),\ldots,(g-1)]$ we take the chain $C_{a,b}=\{(f(t,\sgn(\alpha_1),f(t,\sgn(\alpha_2))+a,f(t,\sgn(\alpha_3))+b: t\in [g]\}\cap [g]^3$ where $f(t,x) = t$ if $x=1$ and $g+1-t$ otherwise.
Intuitively, function $f$ ensures that the specific coefficient in the chain is either increasing or decreasing, depending on the sign of the corresponding coefficient of $\alpha$.
Then any triple $(i,j,k)\in[g]^3$ is covered by the chain $C_{a,b}$ for
$a=j-i$ if $\sgn(\alpha_1)=\sgn(\alpha_2)$ or $a=j+i-(g+1)$ if $\sgn(\alpha_1)\ne\sgn(\alpha_2)$ and $b$ calculated in the same way.

This concludes the proof that $|R|=\Oh(g^2)$ and hence the running time is $\Ohtilda(n+g^2)$.
\end{proof}
\noindent
This immediately implies:

\begin{corollary}[Corollary 13 \cite{FischerKP24}]\label{cor:det_3sum_c13}
 For any $g\geq 1$, a given $3\LDT_?(3,\bar\alpha,t)$ instance on sets of $n$ numbers over the universe $[-U,U]$ can be deterministically reduced to $\Oh(g^2)$ instances on sets of at most $\Oh(n/g)$ numbers over the universe $[-U',U']$ where $U'=\Oh(U/g)$.
 The running time of the reduction is $\Ohtilda(ng)$.
\end{corollary}

\begin{observation}[Observation 15]\label{obs:det_3sum_o15}
 Let $U\geq U'$.
 An instance of $3\LDT_?(3,\bar\alpha,t)$ on sets of $n$ numbers over $[-U,U]$ can be reduced to $\Oh((\frac{U}{U'})^3)$ many instances of $3\LDT_?(3,\bar\alpha,t)$ over $[-U',U']$ on sets of at most $n$ numbers each.
\end{observation}
\begin{proof}
Exactly the same as the original one.
\end{proof}
\noindent
Now we show how to extend the proof of Theorem 2 from \cite{FischerKP24} from 3\SUM{} to 3\LDT{}.
Let $\delta = \eps/32, \mu = 2-2\delta$ and $\alpha =\frac12+2\delta$.
At a high level, the authors of \cite{FischerKP24}:
\begin{enumerate}
 \item apply Lemma~\ref{le:det_3sum_l11} to find a modulus $m$ that gives few pseudo-solutions (triples $x_1,x_2,x_3\in A_1\times A_2\times A_3$ such that $\sum_i\alpha_ix_i \equiv t \pmod m$),
 \item reduce listing pseudo-solutions to listing solutions,
 \item apply self-reduction from Corollary~\ref{cor:det_3sum_c13} to obtain instances on sets of at most $n^{1-\alpha}$ numbers over $[-U',U']$ for $U'=\Oh(m/n^\alpha)=\Oh(n^{\mu-\alpha})$,
 \item apply reduction from Observation~\ref{obs:det_3sum_o15} to further reduce the universe size to $\Oh(n^{3-3\alpha})$.
\end{enumerate}
\noindent
We already showed how to adapt steps 1,3 and 4 to the reduction for 3\LDT.
For step 2, the authors use the 1-to-$\Oh(1)$ correspondence between pseudo-solutions within set $A$ and solutions within set $A'=\{a \bmod m, (a\bmod m)+m: a\in A\}$.
In the case of 3\LDT{} this is not sufficient, because $\sum_i \alpha_i x_i \equiv t \pmod m$ if and only if $\sum_i \alpha_i (x_i \bmod m) - t = m\cdot k$ for $k\in \{-r,\ldots,r\}$ where $r= 1+ \sum_i|\alpha_i|$.

Let $d=\gcd(x_1,x_2,x_3)$.
We apply the Chinese remainder theorem to calculate integer triple $\bar y$ such that $\sum_i \alpha_i y_i=\gcd(x_1,x_2,x_3)=d$.
As coefficients $\bar\alpha$ and $t$ are non-trivial, we have $d | t$.
For every $k\in \{-r,\ldots,r\}$ we have $d|mk$, so we can write $\sum_i \alpha_i (x_i \bmod m) - t = \frac{mk}{d}(\sum_i \alpha_iy_i)$ which is equivalent to
$$\sum_i \alpha_i \left(x_i \bmod m  -\frac{mk}{d}y_i\right) = t $$

This gives us a constant number of instances of $3\LDT$ over $[-cm,cm]$ for $c=\Oh(1)$ to consider. Then the original instance $A$ is a yes-instance if and only if at least on of the created instances has a valid solution or step 1 searching for the good modulus $m$ returned that $A$ is a yes-instance.
The constant $c$ is subsumed during the step 3 and the analysis of the remaining steps follows.
\end{proof}

Recall Definition~\ref{def:reduction} of subquadratic reduction that we denote with $\red$.
In our case the size of input is $\Oh(n\log U)$, and we consider universes of size at most $2^{n^{o(1)}}$.
Then $\Oh(n^{2-\eps'}\log^{c}U)$ is subquadratic in $\Oh(n\log U)$, so the above theorem gives us a subquadratic reduction:

\begin{corollary}\label{cor:universe-reduction}
For all non-trivial coefficients $\bar\alpha$ and $t$:
 $$3\LDT_?(3,\bar\alpha,t) \red 3\LDT_3(3,\bar\alpha,t).$$
\end{corollary}

Now we can combine this result with the existing reductions from Section~\ref{se:3LDT} to obtain a similar claim for 1-partite 3\LDT:

\begin{lemma}
For all non-trivial coefficients $\bar\alpha$ and $t$:
$$ 3\LDT_?(1,\bar\alpha,t)
\red 3\LDT_3(1,\bar\alpha,t)$$
% For all non-trivial coefficients $\bar\alpha$ and $t$, if an instance of $\LDT_3(1,\bar\alpha,t)$ of $n$ numbers over $[-n^3,n^3]$ can be solved in deterministic time $\Oh(n^{2-\eps})$ for some $\eps>0$, then $\LDT_?(1,\bar\alpha,t)$ on $n$ numbers from $[-U,U]$ can be solved in deterministic time $\Oh(n^{2-\eps'}\log^cU)$ for some constants $\eps',c>0$.
\end{lemma}
\begin{proof}

The reduction combines three existing reductions:
$$ 3\LDT_?(1,\bar\alpha,t)
\stackrel{\text{L~\ref{lm:from_1_to_3}}}\red 3\LDT_?(3,\bar\alpha,t)
\stackrel{\text{C~\ref{cor:universe-reduction} }}\red 3\LDT_3(3,\bar\alpha,t)
\stackrel{\text{T~\ref{thm:everything-equivalent}}}\red 3\LDT_3(1,\bar\alpha,t)$$
\noindent
In the first the step we can use Lemma~\ref{lm:from_1_to_3}, as it does not depend on the values of the elements.
\end{proof}
\noindent
Clearly, $3\LDT_3(p,\bar\alpha,t) \red 3\LDT_?(p,\bar\alpha,t)$, so we obtain the main result of this section:

\AllLDTAtMost*

\noindent
Finally, we can obtain similar universe reductions for Convolution variants of 3\LDT:

\AllConvLDTAtMost*
\begin{proof}
It suffices to show that or all $p\in\{1,3\}$ and non-trivial coefficients $\bar\alpha$ and $t$:
$$\Conv_?(p,\bar\alpha,t) \red \Conv_2(p,\bar\alpha,t)$$

The cases of $p=1$ and $p=3$ are established by the following chain of inequalities that we justify in detail below:
\begin{align*}
\Conv_?(1,\bar\alpha,t)
&\reddesc{T~\ref{thm:conv_1_equiv_3}} \Conv_?(3,\bar\alpha,t) \reddesc{L~\ref{le:3-partite-conv-equivalent}} \Conv_?(3,\bar\alpha,0) \\
&\reddesc{L~\ref{le:conv_to_nonconv}} 3\LDT_?(3,\bar\alpha,0)
\reddesc{C~\ref{cor:universe-reduction}} 3\LDT_3(3,\bar\alpha,0) \\
&\eqqdesc{T~\ref{th:everything-equivalent}} \Conv_2(3,\bar\alpha,t)
\eqqdesc{T~\ref{th:everything-equivalent}} \Conv_2(1,\bar\alpha,t)
\end{align*}

In the first inequality, we can use Theorem~\ref{thm:conv_1_equiv_3}, because it only uses color-coding (Fact~\ref{fact:color-coding}), which does not depend on the values of the elements and Lemma~\ref{le:perp}.
The latter operates on the values which are only a slight, linear modification of the original ones, so the claim can be adjusted.

In the second inequality, we again adapt Lemma~\ref{le:3-partite-conv-equivalent} to bigger values.
Finally, in the third inequality we can use Lemma~\ref{le:conv_to_nonconv} in which we only apply a linear function to each of the elements, so it also applies for numbers from $[-U,U]$.
The claims used in the following steps apply directly, with no further modifications.
Hence the claim follows.
\end{proof}

\bibliographystyle{plainurl}
\bibliography{main}

\end{document}